\documentclass[11pt,leqno]{article}
\usepackage[margin=1.6in]{geometry}

\usepackage{amsmath,amssymb,amsthm}

\newtheorem{theorem}{Theorem}
\newtheorem{lemma}{Lemma}
\newtheorem{proposition}{Proposition}
\newtheorem{corollary}{Corollary}
\theoremstyle{definition}

\theoremstyle{remark}
\newtheorem{claim}{Claim}

\numberwithin{equation}{section}

\DeclareMathOperator{\sol}{\mathsf{sol}}
\DeclareMathOperator{\UCT}{UCT}
\let\uct\UCT

\DeclareMathOperator{\loss}{loss}
\newcommand{\CSP}[1]{\ensuremath{\operatorname{CSP}(#1)}}
\newcommand{\tuple}[1]{ \overline{#1} }

\renewcommand{\vec}[1]{ {\mathbf{#1}} }
\newcommand{\inst}[1]{ {\mathcal{#1}} }

\newcommand{\IPQ}[1]{\text{$\text{(IPQ)}_{#1}$}}

\newcommand{\tuct}{\mathcal{T}}

\DeclareMathOperator{\csp}{CSP}
\DeclareMathOperator{\pr}{pr}
\let\eps\varepsilon \let\epsilon\eps
\newcommand{\calI}{\inst I}
\newcommand{\calC}{{\mathcal C}}
\newcommand{\zd}{,\ldots,}
\newcommand{\bx}{\mathbf{x}}
\newcommand{\by}{\mathbf{y}}

\newcommand{\bu}{\mathbf{u}}
\newcommand{\bv}{\mathbf{v}}

\newcommand{\vprod}[2]{#1  #2}
\newcommand{\vproddot}[2]{#1 \cdot #2}
\newcommand{\Exp}{\mathbb{E}}
\newcommand{\Prob}[1]{\Pr \brc{#1 }}

\newcommand{\VV}[1]{{\cal V}_{#1}}

\newcommand{\brc}[1]{\left(#1\right)}

\usepackage[
  pdftitle={Robust algorithms with polynomial loss for near-unanimity CSPs},
  pdfauthor={Victor Dalmau; Marcin Kozik; Andrei Krokhin; Konstantin Makarychev; Yury Makarychev; Jakub Oprsal},
  hidelinks]{hyperref}

\title{Robust algorithms with polynomial loss for near-unanimity CSPs%
\thanks{Marcin Kozik and Jakub Opr\v{s}al were partially supported by the
  National Science Centre Poland under grant no.\ UMO-2014/13/B/ST6/01812;
  Jakub Opr\v{s}al has also received funding from the European Research Council
  (Grant Agreement no.\ 681988, CSP-Infinity).
  Yury Makarychev was partially supported by NSF awards CAREER CCF-1150062 and
  IIS-1302662.
  V\'{\i}ctor Dalmau was partially supported by MINECO under grant
  TIN2016-76573-C2-1-P and Maria de Maeztu Units of Excellence programme
  MDM-2015-0502.
  A~preliminary version of this paper appeared in SODA 2017.}}

\author{
  V\'{\i}ctor Dalmau\\
    University Pompeu Fabra\and
  Marcin Kozik\\
    Jagiellonian University\and
  Andrei Krokhin\\
    Durham University\and
  Konstantin Makarychev\\
    Northwestern University\and
  Yury Makarychev\\
    TTIC\and
  Jakub Opr\v{s}al\\
    TU Dresden}

\date{}

\begin{document}

  \maketitle

  \begin{abstract} 
    An instance of the Constraint Satisfaction Problem (CSP) is given by a
    family of constraints on overlapping sets of variables, and the goal is to
    assign values from a fixed domain to the variables so that all constraints
    are satisfied. In the optimization version, the goal is to maximize the
    number of satisfied constraints. An approximation algorithm for CSP is
    called robust if it outputs an assignment satisfying a
    $(1-g(\eps))$-fraction of constraints on any $(1-\eps)$-satisfiable
    instance, where the loss function $g$ is such that $g(\eps)\rightarrow 0$ as
    $\eps\rightarrow 0$.

    We study how the robust approximability of CSPs depends on the set of
    constraint relations allowed in instances, the so-called constraint
    language.  All constraint languages admitting a robust polynomial-time
    algorithm (with some $g$) have been characterised by Barto and Kozik, with
    the general bound on the loss $g$ being doubly exponential, specifically
    $g(\eps)=O((\log\log(1/\eps))/\log(1/\eps))$. It is natural to ask when a
    better loss can be achieved: in particular, polynomial loss
    $g(\eps)=O(\eps^{1/k})$ for some constant $k$. In this paper, we consider
    CSPs with a constraint language having a near-unanimity polymorphism. This 
    general condition almost matches a known necessary condition for 
    having a robust algorithm with polynomial loss. We
    give two randomized robust algorithms with polynomial loss for such CSPs:
    one works for any near-unanimity polymorphism and the parameter $k$ in the
    loss depends on the size of the domain and the arity of the relations in
    $\Gamma$, while the other works for a special ternary near-unanimity
    operation called dual discriminator with $k=2$ for any domain size.  In the
    latter case, the CSP is a common generalisation of {\sc Unique Games} with a
    fixed domain and {\sc 2-Sat}.  In the former case, we use the algebraic
    approach to the CSP.  Both cases use the standard semidefinite programming
    relaxation for CSP.
  \end{abstract}

  \section{Introduction}\label{sec:intro}

    The constraint satisfaction problem (CSP) provides a framework in which it
    is possible to express, in a natural way, many combinatorial problems
    encountered in computer science and
    AI~\cite{Cohen06:handbook,Creignou01:book,Feder98:monotone}. An instance of
    the CSP consists of a set of variables, a domain of values, and a set of
    constraints on combinations of values that can  be taken by certain subsets
    of variables. The basic aim is then to find an assignment of values to the
    variables that satisfies the constraints (decision version) or that
    satisfies the maximum number of constraints (optimization version).  

    Since CSP-related algorithmic tasks are usually hard in full generality, a
    major line of research in CSP studies how possible algorithmic solutions
    depend on the set of relations allowed to specify constraints, the so-called
    {\em constraint language}, (see,
    e.g.~\cite{Bulatov??:classifying,Cohen06:handbook,Creignou01:book,Feder98:monotone,Krokhin17:book}).
    The constraint language is denoted by $\Gamma$ and the corresponding CSP by
    $\CSP\Gamma$.  For example, when one is interested in polynomial-time
    solvability (to optimality, for the optimization case), the ultimate sort of
    results are dichotomy
    results~\cite{Bulatov17:dichotomy,Bulatov??:classifying,Feder98:monotone,Kolmogorov17:complexity,Thapper16:finite,Zhuk17:dichotomy},
    pioneered by~\cite{Schaefer78:complexity}, which characterise the tractable
    restrictions and show that the rest are NP-hard. Classifications with
    respect to other complexity classes or specific algorithms are also of
    interest
    (e.g.~\cite{Barto14:jacm,Barto12:NU,Kolmogorov15:power,Larose09:universal}).
    When approximating (optimization) CSPs, the goal is to improve, as much as
    possible, the quality of approximation that can be achieved in polynomial
    time, see e.g.\ surveys~\cite{Khot10:UGCsurvey,MM17:cspsurvey}.
    Throughout the paper we assume that P$\ne$NP.

    The study of {\em almost satisfiable} CSP instances features prominently in
    the approximability literature.  On the hardness side, the notion of
    approximation resistance (which, intuitively, means that a problem cannot be
    approximated better than by just picking a random assignment, even on almost
    satisfiable instances) was much studied recently,
    e.g.~\cite{Austrin13:usefulness,Chan13:resist,Hastad14:maxnot2,Khot14:strong}.
    Many exciting developments in approximability in the last decade were driven
    by the {\em Unique Games Conjecture} (UGC) of Khot, see
    survey~\cite{Khot10:UGCsurvey}.  The UGC states that it is NP-hard to tell
    almost satisfiable instances of $\CSP\Gamma$ from those where only a small
    fraction of constraints can be satisfied, where $\Gamma$ is the constraint
    language consisting of all graphs of permutations over a large enough
    domain.  This conjecture (if true) is known to imply optimal
    inapproximability results for many classical optimization
    problems~\cite{Khot10:UGCsurvey}. Moreover, if the UGC is true then a simple
    algorithm based on semidefinite programming (SDP) provides the best possible
    approximation for all optimization problems
    $\CSP\Gamma$~\cite{Prasad08:optimal}, though the exact quality of this
    approximation is unknown.

    On the positive side, Zwick~\cite{Zwick98:finding} initiated the systematic
    study of approximation algorithms which, given an almost satisfiable
    instance, find an almost satisfying assignment. Formally, call a
    polynomial-time algorithm for CSP {\em robust} if, for every $\eps>0$ and
    every $(1-\eps)$-satisfiable instance (i.e., at most a $\eps$-fraction of
    constraints can be removed to make the instance satisfiable), it outputs a
    $(1-g(\eps))$-satisfying assignment (i.e., that fails to satisfy at most a
    $g(\eps)$-fraction of constraints). Here, the {\em loss} function $g$ must
    be such that $g(\eps)\rightarrow 0$ as $\eps\rightarrow 0$. Note that one
    can without loss of generality assume that $g(0)=0$, that is, a robust
    algorithm must return a satisfying assignment for any satisfiable instance.
    The running time of the algorithm should not depend on $\eps$ (which is
    unknown when the algorithm is run). Which problems $\CSP\Gamma$ admit robust
    algorithms? When such algorithms exist, how does the best possible loss $g$
    depend on $\Gamma$?

    \subsection*{Related Work}
    In~\cite{Zwick98:finding}, Zwick gave an SDP-based robust algorithm with
    $g(\eps)=O(\eps^{1/3})$ for {\sc 2-Sat} and an LP-based robust algorithm with
    $g(\eps)=O(1/\log(1/\eps))$ for {\sc Horn $k$-Sat}.  Robust algorithms with
    $g(\eps)=O(\sqrt{\eps})$ were given in~\cite{Charikar09:near} for {\sc
    2-Sat}, and in~\cite{Charikar06:near} for {\sc Unique Games($q$)} where $q$
    denotes the size of the domain.  For {\sc Horn-2-Sat}, a robust algorithm
    with $g(\eps)=2\eps$ was given in~\cite{Guruswami12:tight}.  These bounds
    for {\sc Horn $k$-Sat} ($k\ge 3$), {\sc Horn $2$-Sat}, {\sc 2-Sat}, and {\sc
    Unique Games($q$)} are known to be
    optimal~\cite{Guruswami12:tight,Khot02:power,Khot07:optimal}, assuming the
    UGC.

    The algebraic approach to
    CSP~\cite{Bulatov??:classifying,Cohen06:handbook,Jeavons97:closure} has
    played a significant role in the recent massive progress in understanding
    the landscape of complexity of CSPs. The key to this approach is the notion
    of a {\em polymorphism}, which is an $n$-ary operation (on the domain) that
    preserves the constraint relations. Intuitively, a polymorphism provides a
    uniform way to combine $n$ solutions to a system of constraints (say, part
    of an instance) into a new solution by applying the operation
    component-wise.  The intention is that the new solution improves on the
    initial solutions in some problem-specific way.  Many classifications of
    CSPs with respect to some algorithmic property of interest begin by proving
    an algebraic classification stating that every constraint language either
    can simulate (in a specific way, via gadgets, -- see
    e.g.~\cite{Barto16:robust,Dalmau13:robust,Larose09:universal} for details)
    one of a few specific basic CSPs failing the property of interest or else
    has polymorphisms having certain nice properties (say, satisfying nice
    equations). Such polymorphisms are then used to obtain positive results,
    e.g.\ to design and analyze algorithms.  Getting such a positive result in
    full generality in one step is usually hard, so (typically) progress is made
    through a series of intermediate steps where the result is obtained for
    increasingly weaker algebraic conditions.  The algebraic approach was
    originally developed for the decision
    CSP~\cite{Bulatov??:classifying,Jeavons97:closure}, and it was adapted for
    robust satisfiability in~\cite{Dalmau13:robust}.

    One such algebraic classification result~\cite{Larose09:ability} gives an
    algebraic condition (referred to as $\mathrm{SD}(\wedge)$ or ``omitting
    types \textbf{1} and \textbf{2}'' --
    see~\cite{Barto14:jacm,Kozik15:maltsev,Larose09:ability} for details)
    equivalent to the \underline{in}ability to simulate {\sc 3-Lin-$p$} -- systems
    of linear equations over $Z_p$, $p$ prime, with 3 variable per equation.
    H\aa stad's celebrated result~\cite{Hastad01:optimal} implies that {\sc
    3-Lin-$p$} does not admit a robust algorithm (for any $g$). This result
    carries over to all constraint languages that can simulate (some) {\sc
    3-Lin-$p$}~\cite{Dalmau13:robust}. The remaining languages are precisely those
    that have the logico-combinatorial property of CSPs called ``{\em bounded
    width}'' or ``{\em bounded treewidth
    duality}''~\cite{Barto14:jacm,Bulatov09:bounded,Larose06:bounded}. This
    property says, roughly, that all unsatisfiable instances can be refuted via
    local propagation -- see~\cite{Bulatov08:duality} for a survey on dualities
    for CSP.  Barto and Kozik used $\mathrm{SD}(\wedge)$ in~\cite{Barto14:jacm},
    and then in~\cite{Barto16:robust} they used their techniques
    from~\cite{Barto14:jacm} to prove the Guruswami-Zhou
    conjecture~\cite{Guruswami12:tight} that each bounded width CSP admits a
    robust algorithm.

    The general bound on the loss in~\cite{Barto16:robust} is
    $g(\eps)=O((\log\log(1/\eps))/\log(1/\eps))$. It is natural to ask when a
    better loss can be achieved. In particular, the problems of characterizing
    CSPs where linear loss $g(\eps)=O(\eps)$ or polynomial loss
    $g(\eps)=O(\eps^{1/k})$ (for constant $k$) can be achieved have been posed
    in~\cite{Dalmau13:robust}. Partial results on these problems appeared
    in~\cite{Dalmau13:robust,Dalmau15:towards,Kun12:robust}. For the Boolean
    case, i.e., when the domain is $\{0,1\}$, the dependence of loss on $\Gamma$
    is fully classified in~\cite{Dalmau13:robust}.

    \subsection*{Our Contribution}
    We study CSPs that admit a robust algorithm with polynomial loss. As
    explained above, the bounded width property is necessary for admitting any
    robust algorithm. {\sc Horn 3-Sat} has bounded width, but does not admit a
    robust algorithm with polynomial loss (unless the UGC
    fails)~\cite{Guruswami12:tight}.  The algebraic condition that separates
    {\sc 3-Lin-$p$} and {\sc Horn 3-Sat} from the CSPs that can potentially be
    shown to admit a robust algorithm with polynomial loss is known as
    $\mathrm{SD}(\vee)$ or ``omitting types \textbf{1}, \textbf{2} and
    \textbf{5}''~\cite{Dalmau13:robust}, see Section~\ref{sec:algebra} for the
    description of $\mathrm{SD}(\vee)$ in terms of polymorphisms. The condition
    $\mathrm{SD}(\vee)$ is also a necessary condition for the
    logico-combinatorial property of CSPs called ``{\em bounded pathwidth
    duality}'' (which says, roughly, that all unsatisfiable instances can be
    refuted via local propagation in a linear fashion), and possibly a
    sufficient condition for it too~\cite{Larose09:universal}. It seems very hard 
    to obtain a robust algorithm with polynomial loss for every CSP satisfying $\mathrm{SD}(\vee)$
    all in one step.

    From the algebraic perspective, the most general natural condition that is
    (slightly) stronger than  $\mathrm{SD}(\vee)$ is the {\em near-unanimity
    (NU)} condition~\cite{Baker75:chinese-remainder}.  CSPs with a constraint
    language having an NU polymorphism received a lot of attention in the
    literature (e.g.~\cite{Feder98:monotone,Jeavons98:consist,Barto12:NU}).
    Bounded pathwidth duality for CSPs admitting an NU polymorphism was
    established in a series of
    papers~\cite{Dalmau05:linear,Dalmau08:majority,Barto12:NU},  and we use some
    ideas from~\cite{Dalmau08:majority,Barto12:NU} in this paper. 
    
    We prove that any CSP with a constraint language having an NU polymorphism
    admits a randomized robust algorithm with loss $O(\eps^{1/k})$, where $k$
    depends on the size of the domain. It is an open question whether this
    dependence on the size of the domain is necessary. We prove that, for the
    special case of a ternary NU polymorphism known as {\em dual discriminator}
    (the corresponding CSP is a common generalisation of {\sc Unique Games} with
    a fixed domain and {\sc 2-Sat}), we can always choose $k=2$. Like the vast majority of approximation algorithms for CSPs~\cite{MM17:cspsurvey}, our algorithms
    use the standard SDP relaxation.

    The algorithm for the general NU case follows the same general scheme as~\cite{Barto16:robust,Kun12:robust}:
    \begin{enumerate}
      \item Solve the LP/SDP relaxation for a $(1-\eps)$-satisfiable instance $\calI$.
      \item Use the LP/SDP solution to remove certain constraints in $\calI$ with
      total weight $O(g(\eps))$ (in our case, $O(\eps^{1/k})$) so that the
      remaining instance satisfies a certain consistency condition.
      \item Use the appropriate polymorphism (in our case, NU) to show
      that any instance of $\CSP\Gamma$ with this consistency condition is
      satisfiable.
    \end{enumerate}
    Steps 1 and 2 in this scheme can be applied to any CSP instance, and this is
    where essentially all work of the approximation algorithm happens.
    Polymorphisms are not used in the algorithm, they are used in step 3 only to
    prove the correctness. While the above general scheme is rather simple, 
    applying it is typically quite challenging. Obviously, step 2 prefers weaker  conditions
    (achievable by removing not too many constraints), while step 3 prefers
    stronger conditions (so that they can guarantee satisfiability), so reaching
    the balance between them is the main (and typically significant) technical challenge in any application
    of this scheme. Our algorithm is somewhat inspired
    by~\cite{Barto16:robust}, but it is also quite different from the algorithm there. 
    That algorithm is designed so
    that steps 1 and 2 establish a consistency condition that, in particular,
    includes the 1-minimality condition, and establishing 1-minimality alone
    requires removing constraints with total weight
    $O(1/\log{(1/\eps)})$~\cite{Guruswami12:tight}, unless UGC fails. Since our requirement on the loss function $g(\eps)$ is stricter, we have to design a different ``rounding'' procedure (which is usually the hardest part to analyse for most approximation algorithms).
As in~\cite{Barto16:robust}, our rounding is non-traditional, since a solution to the SDP relaxation is used to decide which constraints to violate, rather than to immediately assign values to the variables. 
 To show that our rounding gives the right dependency on $\eps$, we introduce a new consistency
    condition somewhat inspired by~\cite{Barto12:NU,Kozik16:circles}.  The proof
    that the new consistency condition satisfies the requirements of steps 2 and
    3 of the above scheme is one of the main technical contributions of our
    paper.

    \subsubsection*{Organization of the paper}

    After some preliminaries, we formulate the two main results of this paper in
    Section \ref{sec:main}. Section \ref{sec:SDP} then contains a~description of
    SDP relaxations that we will use further on. Sections \ref{sec:overview1}
    and \ref{sec:overview-thm2} contain the description of the algorithms for
    constraint languages compatible with NU polymorphism and dual discriminator,
    respectively; the following chapters prove the correctness of the two
    algorithms.

  \section{Preliminaries}

    \subsection{CSPs}

    Throughout the paper, let $D$ be a {\em fixed finite} set, sometimes called
    the {\em domain}.  An {\em instance} of the $\csp$ is a pair
    $\calI=(V,{\mathcal C})$ with $V$ a finite set of {\em variables} and
    ${\mathcal C}$ is a finite set of constraints. Each constraint is a pair
    $(\overline{x},R)$ where $\overline{x}$ is a tuple of variables (say, of
    length $r>0$), called the {\em scope} of $C$ and $R$ an $r$-ary relation on
    $D$ called the {\em constraint relation} of $C$. The arity of a constraint
    is defined to be the arity of its constraint relation. In the weighted
    optimization version, which we consider in this paper, every constraint
    $C\in{\mathcal C}$ has an associated {\em weight} $w_C\geq 0$.  Unless
    otherwise stated we shall assume that every instance satisfies
    $\sum_{C\in{\mathcal C}} w_C=1$.

    An {\em assignment} for $\calI$ is a mapping $s\colon V\rightarrow D$. We say that
    $s$ satisfies a constraint $((x_1,\dots,x_r),R)$ if
    $(s(x_1),\dots,s(x_r))\in R$. 
    For $0\leq \beta\leq 1$ we say that assignment $s$ $\beta$-satisfies $\calI$
    if the total weight of the constraints satisfied by $s$ is at least $\beta$.
    In this case we say that $\calI$ is $\beta$-satisfiable. The best possible
    $\beta$ for $\calI$ is denoted by $\mathrm{Opt}(\calI)$.

    A {\em constraint language} on $D$ is a {\em finite} set $\Gamma$ of
    relations on $D$. The problem $\csp(\Gamma)$ consists of all instances of
    the CSP where all the constraint relations are from $\Gamma$.  Problems {\sc
    $k$-Sat}, {\sc Horn $k$-Sat}, {\sc 3-Lin-$p$}, {\sc Graph $H$-colouring}, and
    {\sc Unique Games$(|D|)$} are all of the form $\CSP\Gamma$.

    The {\em decision problem} for $\csp(\Gamma)$ asks whether an input instance
    $\calI$ of $\csp(\Gamma)$ has an assignment satisfying all constraints in
    $\calI$. The {\em optimization problem} for $\csp(\Gamma)$ asks to find an
    assignment $s$ where the weight of the constraints satisfied by $s$ is as
    large as possible. Optimization problems are often hard to solve to
    optimality, motivating the study of {\em approximation} algorithms.

    \subsection{Algebra}\label{sec:algebra}

    An $n$-ary operation $f$ on $D$ is a map from $D^n$ to $D$. We say that $f$
    {\em preserves} (or is a {\em polymorphism} of) an $r$-ary relation $R$ on
    $D$ if for all $n$ (not necessarily distinct) tuples $(a^i_1,\dots,a_r^i)\in
    R$, $1\leq i\leq n$, the tuple
    \[
      (f(a_1^1,\dots,a_n^1),\dots,f(a_1^r,\dots,a_n^r))
    \]
    belongs to $R$ as well.
    Say, if $R$ is the edge relation of a digraph $H$, then $f$ is a
    polymorphism of $R$ if and only if, for any list of $n$ (not necessarily
    distinct) edges $(a_1,b_1),\ldots,(a_n,b_n)$ of $H$, there is an edge in $H$
    from $f(a_1,\ldots,a_n)$ to $f(b_1,\ldots,b_n)$.
    If $f$ is a polymorphism
    of every relation in a constraint language $\Gamma$ then $f$ is called a
    polymorphism of $\Gamma$.  Many algorithmic properties of $\CSP\Gamma$
    depend only on the polymorphisms of
    $\Gamma$, see survey~\cite{Barto17:poly}, also~\cite{Bulatov??:classifying,Dalmau13:robust,Jeavons97:closure,Larose09:universal}.

    An $(n+1)$-ary ($n\ge 2$) operation $f$ is a {\em near-unanimity (NU)} operation
    if, for all $x,y\in D$, it satisfies
    \begin{multline*}
      f(x,x,\ldots,x,x,y)=f(x,x,\ldots,x,y,x)=\dots
      =f(y,x,\ldots,x,x,x)=x.
    \end{multline*}
    Note that the behaviour of $f$ on other tuples of arguments is not
    restricted.
    An NU operation of arity 3 is called a {\em majority} operation. 

    We mentioned in the introduction that (modulo UGC) only constraint languages
    satisfying condition $\mathrm{SD}(\vee)$ can admit robust algorithms with
    polynomial loss. The condition $\mathrm{SD}(\vee)$ can be expressed in many
    equivalent ways: for example, as the existence of ternary polymorphisms
    $d_0\zd d_t$, $t\ge 2$, satisfying the following
    equations~\cite{Hobby88:structure}:
    \begin{align}
      d_0(x,y,z) &= x,\quad d_t(x,y,z) = z, \\
      d_i(x,y,x) &= d_{i+1}(x,y,x) \text{  for all even $i<t$}, \label{sdjoin3}\\
      d_i(x,y,y) &= d_{i+1}(x,y,y) \text{  for all even $i<t$},\\
      d_i(x,x,y) &= d_{i+1}(x,x,y) \text{  for all odd $i<t$}.
    \end{align}
    If line (\ref{sdjoin3}) is strengthened to $d_i(x,y,x)=x$ for all $i$, then,
    for any constraint language, having such polymorphisms would be equivalent
    to having an NU polymorphism of some arity~\cite{Barto13:NU} (this is true only when 
    constraint languages are assumed to be finite).
    
    NU polymorphisms appeared many times in the CSP literature. For example,
    they characterize the so-called ``bounded strict width''
    property~\cite{Feder98:monotone,Jeavons98:consist}, which says, roughly,
    that, after establishing local consistency in an instance, one can always
    construct a solution in a greedy way, by picking values for variables in any
    order so that constraints are not violated. 
    
    \begin{theorem}\cite{Feder98:monotone,Jeavons98:consist}
    \label{the:maj} \label{the:nu}
    Let $\Gamma$ be a constraint language with an NU polymorphism of some arity.
    There is a polynomial-time algorithm that, given an instance of
    $\csp(\Gamma)$, finds a satisfying assignment or reports that none exists.
    \end{theorem}

    Every
    relation with an $(n+1)$-ary NU polymorphism is {\em $n$-decomposable} (and
    in some sense the converse also holds)~\cite{Baker75:chinese-remainder}. We
    give a formal definition only for the majority case $n=2$. Let $R$ be a
    $r$-ary ($r\ge 2$) relation.  For every $i,j\in\{1,\dots,r\}$, let
    $\pr_{i,j} R$ be the binary relation $\{(a_i,a_j)\mid (a_1,\dots,a_r)\in
    R\}$. Then $R$ is called $2$-{\em decomposable} if the following holds: a
    tuple $(a_1,\dots,a_r)\in D^r$ belongs to $R$ if and only if  $(a_i,a_j)\in
    \pr_{i,j} R$ for every $i,j\in\{1,\dots,r\}$.

    The {\em dual discriminator} is a majority operation $f$ such that
    $f(x,y,z)=x$ whenever $x,y,z$ are pairwise distinct. Binary relations
    preserved by the dual discriminator are known as {\em
    implicational}~\cite{Bibel88:deductive} or {\em
    0/1/all}~\cite{Cooper94:characterising} relations. Every such relation is of one of the four following types:
    \begin{enumerate}
      \item $(\{a\}\times D)\cup (D\times\{b\})$ for $a,b\in D$,
      \item $\{(\pi(a),a)\mid a\in D\}$ where $\pi$ is a permutation on $D$,
      \item $P\times Q$ where $P,Q\subseteq D$,
      \item a intersection of a relation of type 1 or 2 with a relation of type 3.
    \end{enumerate}
    The relations of the first kind, when $D=\{0,1\}$, are exactly the relations
    allowed in {\sc 2-Sat}, while the relations of the second kind are precisely
    the relations allowed in {\sc Unique Games $(|D|)$}.  We remark that having
    such an explicit description of relations having a given polymorphism is
    rare beyond the Boolean case.

  \section{Main result} \label{sec:main}

    \begin{theorem} \label{the:main}
    Let $\Gamma$ be a constraint language on $D$.
    \begin{enumerate}
      \item If $\Gamma$ has a near-unanimity polymorphism then $\CSP\Gamma$ admits a
      randomized polynomial-time robust algorithm with loss $O(\eps^{1/k})$ for $k = 6|D|^r +7$
      where $r$ is the maximal arity of a~relation in $\Gamma$.
      Moreover, if $\Gamma$ contains only binary relations then one can choose $k=6|D|+7$.
      \item If $\Gamma$ has the dual discriminator polymorphism then
      $\CSP\Gamma$ admits a randomized polynomial-time robust algorithm with loss
      $O(\sqrt{\eps})$.
    \end{enumerate}
    \end{theorem}

    It was stated as an open problem in~\cite{Dalmau13:robust} whether every CSP
    that admits a robust algorithm with loss $O(\eps^{1/k})$ admits one where
    $k$ is bounded by an absolute constant (that does not dependent on $D$). In
    the context of the above theorem, the problem can be made more specific: is
    dependence of $k$ on $|D|$ in this theorem avoidable or there is a strict
    hierarchy of possible degrees there? The case of a majority polymorphism is 
    a good starting point when trying to answer this question.

    As mentioned in the introduction, robust algorithms with polynomial loss and
    bounded pathwidth duality for CSPs seem to be somehow related, at least in
    terms of algebraic conditions. The condition $\mathrm{SD}(\vee)$ is the
    common necessary condition for them, albeit it is conditional on UGC for the
    former  and unconditional for the latter. Having an NU polymorphism is a
    sufficient condition for both.  Another family of problems $\CSP\Gamma$ with
    bounded pathwidth duality was shown to admit robust algorithms with
    polynomial loss in~\cite{Dalmau13:robust}, where the parameter $k$ depends
    on the pathwidth duality bound (and appears in the algebraic description of
    this family). This family includes languages not having an NU polymorphism
    of any arity -- see~\cite{Carvalho10:caterpillar,Carvalho11:lattice}. It is
    unclear how far connections between the two directions go, but consistency
    notions seem to be the common theme. 
    
    Returning to the discussion of a possible hierarchy of degrees in polynomial
    loss in robust algorithms -- there was a similar question about a hierarchy
    of bounds for pathwidth duality, and the hierarchy was shown to be
    strict~\cite{Dalmau08:majority}, even in the presence of a majority
    polymorphism.

    \section{SDP relaxation}\label{sec:SDP}

    Associated to every instance $\calI=(V,{\mathcal C})$ of CSP there is a
    standard SDP relaxation. It comes in two versions: maximizing the number of
    satisfied constraints and minimizing the number of unsatisfied constraints.
    We use the latter.  We define it assuming that all constraints are binary, this will be sufficient
    for our purposes.
    The SDP has a variable $\bx_a$ for every $x\in V$ and $a\in D$. It also
    contains a special unit vector $\bv_0$. The goal is to assign
    $(|V\|D|)$-dimensional real vectors to its variables minimizing the
    following objective function:
    \begin{equation}
      \sum_{C=((x,y),R)\in {\mathcal C}} w_C\sum_{(a,b)\not\in R} \bx_a\by_b
      \label{sdpobj}
    \end{equation}
    subject to:
    \begin{align}
      &\bx_a\by_b\geq 0  &  x,y\in V, a,b\in D \label{sdp1} \\
      &\bx_a\bx_b= 0  &  x\in V, a,b\in D, a\neq b \label{sdp2} \\
      &\textstyle\sum_{a\in D} \bx_a=\bv_0  &  x\in V  \label{sdp3} \\
      &\|\bv_0\|=1  &  \label{sdp4}
    \end{align}
	
    In the intended integral solution, $x=a$ if $\bx_a=\bv_0$. In the fractional
    solution, we informally interpret $\|\mathbf{x}_a\|^2$ as the 
    probability
    of
    $x=a$ according to the SDP (the constraints of the SDP ensure that
    $\sum_{a\in D} \|\mathbf{x}_a\|^2 =1$).  If $C=((x,y),R)$ is a constraint
    and $a,b\in D$, one can think of $\bx_a\by_b$ as the probability
    given by the solution of the SDP to the pair $(a,b)$ in $C$. The optimal SDP
    solution, then, gives as little probability as possible to pairs that are not in
    the constraint relation. For a constraint $C=((x,y),R)$, conditions
    (\ref{sdp3}) and (\ref{sdp4}) imply that $\sum_{(a,b)\in R} \bx_a\by_b$ is
    at most $1$. Let $\loss(C)=\sum_{(a,b)\not\in R} \bx_a\by_b$. For a subset
    $A\subseteq D$, let $\bx_A=\sum_{a\in A} \bx_a$. Note that
    $\bx_D=\by_D(=\bv_0)$ for all $x,y\in D$.

    Let $\mathrm{SDPOpt}(\calI)$ be the optimum value of (\ref{sdpobj}).  It is
    clear that, for any instance $\calI$, we have $\mathrm{Opt}(\calI)\ge
    \mathrm{SDPOpt}(\calI)\ge 0$. There are algorithms 
    ~\cite{Vandenberghe96:semidefinite} that, given an SDP instance $\calI$ and some 
    additive
    error $\delta>0$, produce in time $\operatorname{\textit{poly}}\, (|\calI|,
    \log(1/\delta))$ an output vector solution whose value is at most
    $\mathrm{SDPOpt}(\calI)+\delta$. There are several ways to deal with the error 
    $\delta$. In this paper we deal with it by introducing a preprocessing step which will also be needed
    to argue that the  algorithm described in the proof of Theorem~\ref{the:main}(1) runs 
    in polynomial time.

        {\bf Preprocessing step 1.}\quad
        Assume that ${\mathcal C}=\{C_1,\dots,C_m\}$ and that
        $w_{C_1}\ge w_{C_2}\ge \ldots \ge w_{C_m}$. Using the algorithm from
        Theorem~\ref{the:nu}, find the largest $j$ such that the subinstance
        $\inst I_j=(V,\{C_1,\dots,C_j\})$ is satisfiable. If the total weight of
        the constraints in $\inst I_j$ is at least $1-1/m$ then return the
        assignment $s$ satisfying $\inst I_j$ and stop. 

        \begin{lemma} \label{lem:prep1} \label{le:prep1}
        Assume that $\calI$ is $(1-\eps)$-satisfiable.
        If $\epsilon\leq 1/m^2$ then preprocessing step 1 returns an assignment that
        $(1-\sqrt{\epsilon})$-satisfies~$\inst I$.
        \end{lemma}

        \begin{proof}
        Assume $\eps\leq 1/m^2$. Let $i$ be maximum with the property that
        $w_{C_i}>\eps$. It follows that the instance
        $\calI_i=(V,\{C_1,\dots,C_i\})$ is satisfiable since the assignment
        $(1-\eps)$-satisfying $\calI$ must satisfy every constraint with weight
        larger than $\eps$. It follows that $i\leq j$ and, hence, the value of the
        assignment satisfying $\calI_j$ is at least $1-w_{C_{i+1}}-\cdots
        -w_{C_m}\geq 1-mw_{C_{i+1}}\geq 1-m\eps\geq 1-\sqrt{\eps}$.
        \end{proof}

        If the preprocessing step returns an assignment then we are done.  So assume
        that it did not return an assignment. Then we know that $\epsilon\ge
        1/m^2$.  We then solve the SDP relaxation with $\delta=1/m^2$ obtaining a 
        solution with objective value at most $2\epsilon$ which is good enough for our 
        purposes.

  \section{Overview of the proof of Theorem \ref{the:main}(1)}
    \label{sec:overview1}
    
    We assume throughout that $\Gamma$ has a near-unanimity polymorphism of
    arity $n+1$ ($n\ge 2$).
    
    It is sufficient to prove Theorem \ref{the:main}(1) for the case when
    $\Gamma$ consists of binary relations and $k=6|D|+7$. The rest will follow
    by Proposition~4.1 of~\cite{Barto16:robust} (see also Theorem~24 in~\cite{Barto17:poly}), 
    which shows how to reduce the
    general case to constraint languages consisting of unary and binary
    relations in such a way that the domain size increases from $|D|$ to $|D|^r$
    where $r$ is the maximal arity of a~relation in $\Gamma$.  Note that every
    unary constraint $(x,R)$ can be replaced by the binary constraint
    $((x,x),R')$ where $R'=\{(a,a)\mid a\in R\}$.

    Throughout the rest of this section, let $\calI=(V,{\mathcal C})$ be a
    $(1-\eps)$-satisfiable instance of $\CSP\Gamma$.

    \subsection{Patterns and realizations}\label{sect:pattern}
      A~\emph{pattern in $\inst I$} is defined as a~directed multigraph $p$
      whose vertices are labeled by variables of $\inst I$ and edges are labeled
      by constraints of $\inst I$ in such a way that the beginning of an edge
      labeled by $((x,y),R)$ is labeled by $x$ and the end by $y$.
      Two of the
      vertices in $p$ can be distinguished as the \emph{beginning} and the
      \emph{end} of~$p$.  If these two vertices are labeled by variables $x$ and $y$,
      respectively, then we say that $p$ is a~pattern from $x$ to $y$.  

      For two patterns $p$ and $q$ such that the end of $p$ and the beginning of
      $q$ are labeled by the same variable, we define $p+q$ to be the pattern
      which is obtained from the disjoint union of $p$ and $q$ by identifying
      the end of $p$ with the beginning of $q$ and choosing the beginning of
      $p+q$ to be the beginning of $p$ and the end of $p+q$ to be the end of $q$.
      We also define $jp$ to be $p + \dots + p$ where $p$ appears $j$ times.
      A~pattern is said to be a~\emph{path pattern} if the underlying graph is
      an~oriented path with the beginning and the end being the two end vertices
      of the path, and is said to be an~\emph{$n$-tree pattern} if the
      underlying graph is an orientation of a tree with at most $n$ leaves, and
      both the beginning and the end are leaves. A~\emph{path of $n$-trees
      pattern} is then any pattern of the form $t_1+\dots+t_j$ for some
      $n$-tree patterns $t_1,\dots,t_j$. 

      A~\emph{realization of a~pattern}
      $p$ is a~mapping $r$ from the set of
      vertices of $p$ to $D$ such that if $(v_x,v_y)$ is an edge labeled by
      $((x,y),R)$ then $(r(v_x),r(v_y)) \in R$. Note that $r$ does not have to
      map different vertices of $p$ labeled with same variable to the same element in $D$.
      A~\emph{propagation} of a~set $A\subseteq D$ along a pattern $p$ whose
      beginning vertex is $b$ and ending vertex is $e$ is defined as follows.
      For $A\subseteq D$, define $A + p=\{r(e) \mid \text{ $r$ is a realization of
      $p$ with  $r(b) \in A$}\}$.
      Also for a~binary relation $R$ we put $A+R=\{b\mid (a,b)\in R \mbox{ and }
      a\in A\}$.
      Observe that 
      we have  $(A+p)+q = A+(p+q)$.

      Further, assume that we have non-empty sets $D_x^\ell$ where $1\leq
      \ell\leq |D|+1$ and $x$ runs through all variables in an instance $\calI$.
      Let $p$ be a pattern in $\calI$ with beginning $b$ and end $e$. We call a
      realization $r$ of $p$ an \emph{$\ell$-realization} (with respect to the
      family $\{D_x^\ell\}$) if, for any vertex $v$ of $p$ labeled by a~variable
      $x$, we have $r(v)\in D_x^{\ell+1}$. For $A\subseteq D$, define
      \[
        A +^\ell p = \{r(e) \mid r
          \text{ is an $\ell$-realization of $p$ with $r(b) \in A$}\}
      .\]
      Also, for a constraint $((x,y),R)$ or $((y,x),R^{-1})$ and sets
      $A,B\subseteq D$, we write $B=A+^\ell (x,R,y)$ if $B=\{b\in D_y^{\ell+1}
      \mid (a,b)\in R \mbox{ for some } a\in A\cap D_x^{\ell+1}\}$.
    
    \subsection{The consistency notion}
      
      Recall that we assume that $\Gamma$ contains only binary relations. Before
      we formally introduce the new consistency notion, which is the key to our
      result, as we explained in the introduction, we give an example of a
      similar simpler condition. We mentioned before that {\sc 2-Sat} is a
      special case of a CSP that admits an NU polymorphism (actually, the only
      majority operation on $\{0,1\}$).
      There is a textbook consistency condition characterizing satisfiable {\sc
      2-Sat} instances, which can be expressed in our notation as follows: for
      each variable $x$ in a {\sc 2-Sat} instance $\calI$, there is a value
      $a_x$ such that, for any path pattern $p$ in $\calI$ from $x$ to $x$, we
      have $a_x\in \{a_x\}+p$.  
      
      Let $\calI$ be an instance of $\CSP\Gamma$ over a set $V$ of variables. We
      say that $\calI$ satisfies condition \IPQ{n} if the following holds:
      \begin{description}
        \item{\IPQ{n}} For every $y\in V$, there exist non-empty sets $D_y^1
        \subseteq \ldots \subseteq D_y^{|D|}  \subseteq D_y^{|D|+1}=D$ such that
        for any $x\in V$, any $\ell \leq |D|$, any $a\in D_x^\ell$, and any two
        patterns $p,q$ which are paths of $n$-trees in $\inst I$ from $x$ to
        $x$, there exists $j$ such that
        \[
          a \in \{a\} +^{\ell} (j(p+q) + p).
        \]
      \end{description}
      Note that $+$ between $p$ and $q$ is the pattern addition and thus
      independent of $\ell$. Note also that $a$ in the above condition belongs to
      $D_x^\ell$, while propagation is performed by using $\ell$-realizations,
      i.e., inside sets $D_y^{\ell+1}$.
 
      The following theorem states that this consistency notion satisfies the
      requirements of step~3 of the general scheme (for designing robust
      approximation algorithms) discussed in the introduction.
      \begin{theorem}\label{thm:nicelevels}
        Let $\Gamma$ be a constraint language containing only binary relations
        such that $\Gamma$ has an $(n+1)$-ary NU polymorphism. If an instance
        $\calI$ of $\CSP\Gamma$ satisfies \IPQ{n}, then $\inst I$ is satisfiable.
      \end{theorem}
      
      \subsection{The algorithm}
      
        Let $k = 6|D| + 7$.  We provide an algorithm which, given
        a~$(1-\epsilon)$-satisfiable instance $\inst I$ of $\CSP\Gamma$, removes
        $O(\epsilon^{1/k})$ constraints from it  to obtain a~subinstance $\inst
        I'$ satisfying condition (IPQ)$_n$.  It then follows
        from Theorem~\ref{thm:nicelevels} that $\inst I'$ is satisfiable, and
        we can find a satisfying assignment by Theorem~\ref{the:nu}.

      \subsubsection{More preprocessing}
      \label{sec:preprocessing}

	    By Lemma \ref{lem:prep1} we can assume that $\epsilon\geq 1/m^2$.
        We solve the SDP relaxation with error $\delta=1/m^2$ and obtain a~solution
     	$\{\vec x_a\}$ $(x\in V,a\in D)$ whose objective value $\epsilon'$ is at most 
        $2\epsilon$. Let us define $\alpha$ to be $\max\{\epsilon',1/m^2\}$. It
        is clear 
        that $\alpha=O(\epsilon)$. Furthermore, this gives us that
        $1/\alpha\le m^2$.
        This will be
        needed to argue that the main part of the algorithm runs in polynomial
        time.

        Let $\kappa=1/k$ (we will often use $\kappa$ to avoid overloading
        formulas).

        {\bf Preprocessing step 2.}\quad For each $x\in V$ and $1\leq \ell\leq
        |D|+1$, compute sets $D_x^{\ell}\subseteq D$ as follows.  Set
        $D_x^{|D|+1}=D$ and, for $1\leq \ell\leq |D|$, set $D_x^{\ell}=\{a\in
        D\mid \|\vec x_a\|\ge r_{x,\ell}\}$ where $r_{x,\ell}$ is the smallest
        number of the form $r=\alpha^{3\ell \kappa}(2|D|)^{i/2}$, $i\ge 0$
        integer, with $\{b\in D\mid r(2|D|)^{-1/2}\le \|\vec x_b\|<
        r\}=\emptyset$.  It is easy to check that $r_{x,\ell}$ is obtained with
        $i\le |D|$.

        It is clear that the sets $D_x^{\ell}\subseteq D$, $x\in V$, $1\leq \ell\leq
        |D|$, can be computed in polynomial time. 

        The sets $D_x^\ell$ are chosen such that $D_x^\ell$
        contains relatively ``heavy'' elements ($a$'s such that $\|\vec x_a\|^2$
        is large).
        The thresholds are chosen so that there is a~big gap (at
        least by a~factor of $2|D|$) between ``heaviness'' of an element in
        $D_x^\ell$ and outside.

      \subsubsection{Main part}
      \label{sec:algorithm}
        Given the preprocessing is done, we have that $1/\alpha \leq m^2$, and
        we precomputed sets $D_x^\ell$ for all $x\in V$ and $1\leq \ell \leq
        |D|+1$.  The description below uses the number $n$, where $n+1$ is the
        arity of the NU polymorphism of $\Gamma$.
        
        \textbf{Step 0.}\quad Remove every constraint $C$ with $\loss(C) >
        \alpha^{1-\kappa}$.

        \textbf{Step 1.}\quad For every $ 1 \leq \ell \leq |D| $ do the following.
        Pick a~value $r_\ell \in (0,\alpha^{(6\ell +4)\kappa})$ uniformly at
        random. Here we need some notation: for $x,y \in V$ and $A,B\subseteq D$,
        we write $\vec x_A \preceq^\ell \vec y_B$ to indicate that there is
        \textbf{no} integer $j$ such that
        \(
          \| \vec y_B \|^2 < r_\ell +j\alpha^{(6\ell +4)\kappa}
          \leq \| \vec x_A \|^2.
        \)
        Then, remove all constraints $((x,y),R)$ such that there are sets $A,B
        \subseteq D$ with $B = A +^\ell (x,R,y)$ and $\vec
        x_A \not\preceq^\ell \vec y_B$, or with $B = A +^\ell (y,R^{-1},x)$ and
        $\vec y_A \not\preceq^\ell \vec x_B$.

        \textbf{Step 2.}\quad For every $1 \leq \ell \leq |D|$ do the following.
        Let $m_0 = \lfloor \alpha^{-2\kappa} \rfloor$.
        Pick a~value $s_\ell \in \{0,\dots, m_0-1\}$
        uniformly at random. We define $\vec x_A \preceq^{\ell}_w \vec y_B$ to
        mean that there is \textbf{no} integer $j$ such that
        \(
          \|\vec y_B\|^2 < r_\ell + (s_\ell +jm_0)\alpha^{(6\ell +4)\kappa}
          \leq \|\vec x_A\|^2.
        \)
        Obviously, if $\vec x_A \preceq^{\ell} \vec y_B$ then $\vec x_A
        \preceq^{\ell}_w \vec y_B$.
        Now, if $A\subseteq B \subseteq D_x^{\ell+1}$  are such that
        $\|\vec x_B - \vec x_A\|^2 \leq (2n-3)\alpha^{(6\ell +4)\kappa}$ and 
        $\vec x_B \not\preceq^{\ell}_w \vec x_A$,
        then remove all the constraints in which $x$ participates.

        \textbf{Step 3.}\quad For every $1\leq \ell \leq |D|$ do the following.
        Pick $m_\ell = \lceil \alpha ^ {-(3\ell +1)\kappa} \rceil$ unit
        vectors independently uniformly at random. For $x,y\in V$ and $A,B
        \subseteq D$, say that $\vec x_A$ and $\vec y_B$ are \emph{cut} by a
        vector $\vec u$ if the signs of $\vec u \cdot (\vec x_A - \vec
        x_{D\setminus A})$ and $\vec u \cdot (\vec y_B - \vec y_{D\setminus B})$
        differ.  Furthermore, we say that $\vec x_A$ and $\vec y_B$ are
        $\ell$-cut if there are cut by at least one of the chosen $m_\ell$
        vectors.
        For every variable $x$, if there exist subsets $A,B \subseteq D$ such that
        $A \cap D_x^\ell \neq B \cap D_x^\ell$ and the vectors $\vec x_A$ and
        $\vec x_B$ are not $\ell$-cut, then remove all the constraints in which
        $x$ participates.

        \textbf{Step 4.}\quad For every $1 \leq \ell \leq |D|$, remove every
        constraint $((x,y),R)$ such that there are sets $A,B \subseteq D$ with
        $B = A +^\ell (x,R,y)$, and $\vec x_A$ and $\vec y_B$ are $\ell$-cut, or with
        $B = A +^\ell (y,R^{-1},x)$, and $\vec y_A$ and $\vec x_B$ are $\ell$-cut.

        \textbf{Step 5.}\quad For every $1 \leq \ell \leq |D|$ do the following.
        For every variable $x$,
        If $A,B\subseteq D_x^{\ell +1}$
        such that $\| \vec x_B - \vec x_A \|^2 \leq
        (2n-3)\alpha^{(6\ell+4)\kappa}$ and $\vec x_A$ and $\vec x_B$ are
        $\ell$-cut, remove all constraints in which $x$ participates.

        \textbf{Step 6.}\quad By Proposition~\ref{prop:consistence} and
        Theorem~\ref{thm:nicelevels}, the remaining instance $\inst I'$ is
        satisfiable. Use the algorithm given by Theorem \ref{the:nu} to find
        a~satisfying assignment for $\inst I'$. Assign all variables in $\inst
        I$ that do not appear in $\inst I'$ arbitrarily and return the obtained
        assignment for $\inst I$.

        \bigskip
        Note that we chose to define the cut condition based on
        $\bx_A-\bx_{D\setminus A}$, rather than on $ \bx_A$, because the former
        choice has the advantage that $\|\bx_A-\bx_{D\setminus A}\|=1$, which
        helps in some calculations.
        
        In step~0 we remove constraints such that, according to the SDP solution, 
        have a high probability to be violated.
        Intuitively, steps 1 and 2 ensure that the loss in $\|\vec x_A\|$ after
        propagating $A$ by a~path of $n$-trees is not too big. This is achieved
        first by ensuring that by following a~path we do not lose too much
        (step~1) which also gives a~bound on how much we can lose by following
        an~$n$-tree pattern (see Lemma \ref{lem:tree-loss}). Together with the removal 
        of constraints in step 2, this guarantees that
        following a~path of $n$-trees we do not lose too much. This ensures that $\{a\} +^\ell (j(p+q) + p)$ is non-vanishing as
        $j$ increases.
        Steps 3--5 ensure that if $A$ and $B$ are connected by paths of
        $n$-trees in both directions (i.e., $\vec x_A = \vec x_B +^{\ell} p_1$ and $\vec 
        x_B = \vec x_A +^{\ell} p_2$), then $\vec x_A$ and $\vec x_B$ do not differ too 
        much (i.e., $A \cap D_x^\ell = B \cap D_x^\ell$). This is achieved by separating 
        the space
        into cones by cutting it using the $m_\ell$ chosen vectors, removing the
        variables which have two different sets that are not $\ell$-cut (step
        3), and then ensuring that if we follow an edge (step 4), or if we drop
        elements that do not extend to an~$n$-tree (step 5) we do not cross
        a~border to another cone.  This gives us both that the sequence $A_j
        = \{a\} +^\ell (j(p+q) +p)$ stabilizes and that, after it stabilizes,
        $A_j$ contains $a$. This provides condition \IPQ{n} for the remaining instance $\inst I'$. 

        The algorithm runs in polynomial time. Since $D$ is fixed, it is clear
        that the steps 0--2 can be performed in polynomial time. For steps 3--5,
        we also need that $m_\ell$ is bounded by a polynomial in $m$, which holds
        because $\alpha \geq 1/m^2$. 
        
        The correctness of the algorithm is given by Theorem~\ref{thm:nicelevels} and the two following
        propositions whose proof can be found in Section
        \ref{sec:correctness-of-algorithm}. These propositions show that our new
        consistency notion satisfies the requirements of step~2 of the general
        scheme for designing robust approximation algorithms
        discussed in the introduction.

        \begin{proposition} \label{prop:expected-lost-weight}
        The expected total weight of constraints removed by the algorithm is
        $O(\alpha^\kappa)$.
        \end{proposition}

        \begin{proposition} \label{prop:consistence}
        The instance $\inst I'$ obtained after steps 0--5 satisfies the
        condition \IPQ{n} (with the sets $D_x^\ell$ computed by preprocessing
        step 2 in Section \ref{sec:preprocessing}).
        \end{proposition}

\section{Overview of the proof of Theorem~\ref{the:main}(2)}\label{sec:overview-thm2}

Since the~dual discriminator is a~majority operation, every relation in $\Gamma$
is 2-decomposable. Therefore, it follows, e.g.\ from Lemma 3.2
in~\cite{Dalmau13:robust}, that to prove that $\CSP\Gamma$ admits a robust
algorithm with loss $O(\sqrt\eps)$, it suffices to prove this for the case when
$\Gamma$ consists of all unary and binary relations preserved by the dual
discriminator.
Such binary constraints are of one of the four kinds described in
Section~\ref{sec:algebra}. Using this description, it follows from Lemma~3.2
of~\cite{Dalmau13:robust} that it suffices to consider the following three types
of constraints:
\begin{enumerate}
\item Disjunction constraints of the form $x= a \vee y = b$, where $a,b\in D$;
\item Unique game (UG) constraints of the form $x = \pi(y)$, where $\pi$ is any permutation on~$D$;
\item Unary constraints of the form $x \in P$, where $P$ is an arbitrary non-empty subset of $D$.
\end{enumerate}

We present an algorithm that, given a $(1-\eps)$-satisfiable instance
$\calI=(V,\calC)$ of the problem, finds a solution satisfying constraints with
expected total weight $1- O(\sqrt{\eps\log{|D|}})$ (the hidden constant in the
$O$-notation depends neither on $\eps$ nor on $|D|$).

We now give an informal and somewhat imprecise sketch of the algorithm and its
analysis.  We present details in Section~\ref{sec:thm2}. We use the SDP
relaxation from Section~\ref{sec:SDP}.  Let us call the value
$\|\mathbf{x}_a\|^2$ the SDP weight of the value $a$ for variable $x$.

\subsubsection*{Variable Partitioning Step}
The algorithm first solves the SDP relaxation. Then, it partitions all variables
into three groups $\VV{0}$, $\VV{1}$, and $\VV{2}$ using a threshold rounding
algorithm with a random threshold.  If most of the SDP weight for $x$ is
concentrated on one value $a\in D$, then the algorithm puts $x$ in the set
$\VV{0}$ and assigns $x$ the value $a$.  If most of the SDP weight for $x$ is
concentrated on two values $a,b\in D$, then the algorithm puts $x$ in the set
$\VV{1}$ and restricts the domain of $x$ to the set $D_x = \{a,b\}$ (thus we
guarantee that the algorithm will eventually assign one of the values $a$ or $b$
to $x$).  Finally, if the SDP weight for $x$ is spread among 3 or more values,
then we put $x$ in the set $\VV{2}$; we do not restrict the domain for such $x$.
After we assign values to $x\in \VV{0}$ and restrict the domain of $x\in \VV{1}$
to $D_x$, some constraints are guaranteed to be satisfied (say, the constraint
$(x=a)\vee (y=b)$ is satisfied if we assign $x$ the value $a$ and the constraint
$x\in P$ is satisfied if $D_x\subseteq P$). Denote the set of such constraints
by ${\cal C}_s$ and let $\calC'=\calC\setminus\calC_s$.

We then identify a set ${\cal C}_v\subseteq \calC'$ of constraints that we conservatively label as violated. This set includes all constraints in $\calC'$ except those belonging to one
of the following 4 groups:
\begin{enumerate}
\item disjunction constraints $(x = a) \vee (y = b)$ with $x, y \in {\cal V}_1$ and $a\in D_x$, $b\in D_y$;
\item UG constraints $x = \pi (y)$ with $x, y \in {\cal V}_1$ and $D_x = \pi(D_y)$;
\item UG constraints $x = \pi (y)$ with $x, y \in {\cal V}_2$;
\item unary constraints $x \in P$ with $x\in{\cal V}_2$.
\end{enumerate}
Our construction of sets $\VV{0}$, $\VV{1}$, and $\VV{2}$, which is based on randomized threshold rounding, ensures that the expected total weight of constraints in ${\cal C}_v$ is $O(\eps)$ (see Lemma~\ref{lem:preproc-violated}).

The constraints from the 4 groups above naturally form two disjoint sub-ins\-tan\-ces of $\calI$: $\calI_1$ (groups 1 and 2) on the set of variables $\VV{1}$, and $\calI_2$ (groups 3 and 4) on $\VV{2}$. We treat these instances independently as described below.

\subsubsection*{Solving Instance $\calI_1$} The instance $\calI_1$ with the domain of each $x$ restricted to $D_x$ is effectively an instance of
Boolean 2-CSP (i.e., each variable has a 2-element domain and all constraints are binary). A robust algorithm with quadratic loss for this problem was given by Charikar et al.~\cite{Charikar09:near}.  This algorithm finds a solution violating an $O(\sqrt{\varepsilon})$ fraction of all constraints if the optimal solution violates at most $\varepsilon$ fraction of all constraints or $\mathrm{SDPOpt}(\calI_1)\leq \eps$. However, we cannot
apply this algorithm to the instance $\calI_1$ as is. The problem is that the weight of violated constraints in the optimal solution
for $\calI_1$ may be greater than $\omega(\varepsilon)$. Note that the unknown optimal solution for the original instance $\calI$ may assign values to variables $x$ outside of the restricted domain $D_x$, and hence it
is not a feasible solution for $\calI_1$. Furthermore, we do not have a feasible SDP solution for the instance $\calI_1$, since the original SDP solution (restricted to the variables in $\VV{1}$) is not a feasible solution for the Boolean 2-CSP problem (because $\sum_{a\in D_x}\mathbf{x}_a$ is not necessarily equal to $\mathbf{v}_0$ and, consequently, $\sum_{a\in D_x}\|\mathbf{x}_a\|^2$ may be less than 1). Thus, our algorithm first transforms the SDP solution to obtain a feasible solution for $\calI_1$. To this end, it partitions the set of vectors $\{\mathbf{x}_a: x\in \VV{1}, a\in D_x\}$ into two sets $H$ and $\bar{H}$ using a modification of the hyperplane rounding algorithm by Goemans and Williamson~\cite{Goemans95:improved}. In this partitioning, for every variable $x$, one of the two vectors
$\{\mathbf{x}_a: a\in D_x\}$ belongs to $H$ and the other belongs to $\bar H$.
Label the elements of each $D_x$ as $\alpha_x$ and $\beta_x$ so that so that $\mathbf{x}_{\alpha_x}$ is the vector in
$H$ and $\mathbf{x}_{\beta_x}$ is the vector in $\bar H$. For every $x$, we define two new vectors
$\mathbf{\tilde{x}}_{\alpha_x} = \mathbf{x}_{\alpha_x}$ and
$\mathbf{\tilde{x}}_{\beta_x} = \bv_0-\mathbf{x}_{\alpha_x}$. It is not hard to verify that the set
of vectors $\{\mathbf{\tilde{x}}_{a}: x\in\VV{1}, a\in D_x\}$ forms a feasible SDP solution for the
instance $\calI_1$. We show that for each disjunction constraint $C$ in the instance $\calI_1$,
the cost of $C$ in the new SDP solution is not greater than the cost of $C$ in the original SDP solution (see Lemma~\ref{lem:new-SDP-feasible}).
The same is true for all but $O(\sqrt{\eps})$ fraction of UG constraints. Thus, after removing UG constraints
for which the SDP value has increased, we get an SDP solution of cost $O(\varepsilon)$. Using the algorithm~\cite{Charikar09:near} for Boolean 2-CSP, we obtain a solution for $\calI_1$
that violates constraints of total weight at most $O(\sqrt{\eps})$.

\subsubsection*{Solving Instance $\calI_2$} The instance $\calI_2$ may contain
only unary and UG constraints as all disjunction constraints are removed from
$\calI_2$ in the variable partitioning step. 
We run the approximation algorithm for
Unique Games by Charikar et al.~\cite{Charikar06:near} on $\calI_2$ using the
original SDP solution restricted to vectors $\{\mathbf{x}_a: x\in \VV{2}, a\in
D\}$. This is a valid SDP relaxation because in the instance $\calI_2$, unlike
the instance $\calI_1$, we do not restrict the domain of variables $x$ to $D_x$.
The cost of this SDP solution is at most $\eps$. As shown
in~\cite{Charikar06:near},  the weight of constraints violated by the
algorithm~\cite{Charikar06:near} is at most $O(\sqrt{\eps\log |D|})$.

We get the solution for $\calI$ by combining solutions for $\calI_1$ and
$\calI_2$, and assigning values chosen at the variable partitioning 
step to the variables from the set $\VV{0}$.

\section{Proof of Theorem~\ref{thm:nicelevels}}\label{sect:proofof3}

  In this section we prove Theorem~\ref{thm:nicelevels}.  The proof will use
  constraint languages with relations of arity greater than two.  In order to
  talk about such instances we need to extend the definition of a pattern.  Note
  that patterns~(in the sense of Section~\ref{sect:pattern}) are instances~
  (with some added structure) and the realizations of patterns are solutions.
  We use the pattern/instance and solution/realization duality to generalize the
  notion of a pattern.  Moreover we often treat patterns as instances
  and~(whenever it makes sense) instances as patterns.

  We will often talk about path/tree instances; they are defined using the {\em
  incidence multigraph}.  The incidence multigraph of an instance $\inst J$ is bipartite, its
  the vertex set consists of variables and constraints of $\inst J$ (which form the two parts), and if a
  variable $x$ appears $j$ times in a constraint $C$ then the vertices
  corresponding to $x$ and $C$ have $j$ edges between them. 

  An instance is {\em connected} if its incidence multigraph is connected;
  an~instance is a~\emph{tree instance} if it is connected and its incidence
  multigraph has no multiple edges and no cycles.  A~\emph{leaf variable} in a
  tree instance is a~variable which corresponds to a~leaf in the incidence
  multigraph, and we say that two variables are \emph{neighbours} if they appear
  together in a~scope of some constraint (i.e., the corresponding vertices are
  connected by a~path of length 2 in the incidence multigraph).  Note that the incidence multigraph of a
  path pattern in a binary instance~(treated as an instance, as described in the
  first paragraph of this section) is a path, and of  an $n$-tree pattern is a
  tree with $n$ leaves.
  
  The next definition captures, among other things, the connection between the
  pattern~(treated as an instance) and the instance in which the pattern is
  defined.  Let $\inst J_1$ and $\inst J_2$ be two instances over the same
  constraint language. 
An~\emph{(instance) homomorphism} $e\colon \inst J_1 \to
  \inst J_2$ is a~mapping that maps each variable of $\inst J_1$ to a variable of
  $\inst J_2$ and each constraint of $\inst J_1$ to constraint of $\inst J_2$ in
  such a~way that every constraint $((y_1,\dots,y_k),R)$ in $\inst J_1$ is mapped to $((e(y_1),\dots,e(y_k)),R)$.

  Using these new notions, a~path pattern in an instance $\inst I$ (see the
  definition in Section~\ref{sect:pattern}) can alternatively be defined as an
  instance, with beginning and end chosen among the leaf variables, whose
  incidence graph is a path from beginning to end, together with a homomorphism
  into $\inst I$.  Similarly we define a {\em path pattern} in a (not
  necessarily binary) instance $\inst I$ as an instance $\inst J$, with chosen
  beginning/end leaf variables,  whose incidence graph, after removing all the
  other vertices of degree one, is a path from beginning to end, together with
  a~homomorphism $e\colon \inst J\rightarrow\inst I$.  Addition of path patterns
  and propagation are defined in an analogous way as for patterns with binary
  constraints~(see Section~\ref{sect:pattern}).

  For a $k$-ary relation $R$, let $\pr_i(R)=\{a_i\mid
  (a_1,\ldots,a_i,\ldots,a_k)\in R\}$.  A CSP instance $\inst J$ is called {\em
  arc-consistent} in sets $D_x$~($x$ ranges over variables
  of $\inst J$) if, for any variable $x$ and any constraint
  $((x_1,\ldots,x_k),R)$ in $\inst J$, if $x_i=x$ then $\pr_i(R)=D_x$.
  We say that a~CSP instance $\inst J$ satisfies condition {\em (PQ)} in sets $D_x$ if 
  \begin{enumerate} 
    \item $\inst J$ is arc-consistent in these sets and \label{cond:pqac}
    \item for any variable $x$, \label{cond:pqpq}
      any path patterns $p,q$ from $x$ to $x$, and 
      any $a\in D_x$ 
      there exists $j$ such that 
      \( a\in \{a\} + (j(p+q) + p) \).
  \end{enumerate}
  Note that if the instance $\inst J$ is binary then (PQ) implies \IPQ{n} for
  all $n$~(setting $D_x^i=D$ if $i=|D|+1$ and $D_x^i=D_x$ if $i<|D|+1$).

  The following fact, 
  a special case of Theorem A.2 in~\cite{Kozik16:circles}, 
  provides solutions for (PQ) instances.
  \begin{theorem} \label{thm:nice}
    If $\Gamma'$ is a~constraint language with a near-unanimity polymorphism, 
    then every instance of $\CSP{\Gamma'}$ satisfying condition (PQ) is satisfiable.
  \end{theorem}
  
  Finally, a standard algebraic notion has not been defined yet:
  having fixed $\Gamma$ over a set $D$, a subset $A\subseteq D$ is a {\em subuniverse} if, for any
  polymorphism $g$ of $\Gamma$, we have $g(a_1,a_2,\ldots)\in A$ whenever
  $a_1,a_2,\ldots \in A$. For any $S\subseteq D$, the subuniverse {\em generated
  by} $S$ is defined as
  \[
    \{g(a_1,\ldots,a_r)\mid
      r\ge 1, a_1,\ldots,a_r\in S,
    \text{$g$ is an $r$-ary polymorphism of $\Gamma$}
    \} 
  \]

  \subsection{Into the proof}
    We begin the proof of Theorem~\ref{thm:nicelevels}.  We fix a binary
    language $\Gamma$ compatible with an $(n+1)$-ary NU polymorphism and an
    instance $\inst I$ of $\CSP\Gamma$ which satisfies $\IPQ{n}$ with sets
    $D_x^\ell$.  Note that we can assume that all $D_x^\ell$'s are subuniverses.
    If this is not the case, we replace each $D_x^\ell$ with the subuniverse generated by
    it.  It is easy to check that~(after the change) the instance $\calI$ still
    satisfies $\IPQ{n}$ with such enlarged $D_x^\ell$'s. 
    
    For each variable $x$, choose and fix an arbitrary index $i$ such that $D_x^i=D_x^{i+1}$ and call it
    the {\em level} of $x$.  Note that each variable has a level~
    (since the sets $D_x^\ell$ are non-empty and $\ell$ ranges from 1 to $|D|+1$).
    Let $V^i$ denote the set of variables of level $i$ and 
    $V^{< i}, V^{\le i},\dots$ be defined in the natural way.
    
    Our proof of  Theorem~\ref{thm:nicelevels} will proceed by applying
    Theorem~\ref{thm:nice} to $\inst I$ restricted to $V^1$, then to $V^2$ and so on.
    However, in order to obtain compatible solutions,
    we will add constraints to the restricted instances.

  \subsection{The instances in levels}
    Let $\inst I^i$~(for $i\leq |D|$) be the instance defined as follows:
    \begin{enumerate}
      \item $V^i$ is the set of variables of $\inst I^i$;
      \item $\inst I^i$ contains, for every $n$-tree pattern $t$ of
        $\inst I$, the constraint $((x_1,\dots,x_k), R)$ defined in the
        following way: \label{enum:newconstraints}
        let $v_1,\dotsc,v_k$ be all the vertices of $t$ labeled by variables
        from $V^i$, then $x_1,\dots,x_k$ are the labels of $v_1,\dots,v_k$
        respectively and
        \[
          R = \{ (r(v_1),\dots,r(v_k)) \mid r
            \text{ is a $i$-realization of $t$ (i.e., inside sets $D_x^{i+1}$)} \}.
        \]
    \end{enumerate}

This definition has a number of immediate consequences: First, every binary
    constraint between two variables from $V^i$ is present in $\inst I^i$~
    (as it defines a two-element $n$-tree).  Second, note that if some $n$-tree contains a vertex $v_j$ in $V^i$ which is not a leaf then 
by splitting the tree $t$ at
    $v_j$~(with $v_j$ included in both parts) we obtain two trees defining constraints which together are equivalent to the constraint defined by $t$. 
    This implies that by including only the constraints defined by $n$-trees $t$ such that only the leaves can be from $V^i$, we obtain an equivalent~(i.e., having the same
    set of solutions) instance.  Throughout most of the proof we will be working
    with such a restricted instance.  In this instance the arity of constraints
    is bounded by $n$.

Since the arity of a constraint in $\inst I^i$ is bounded and
    the size of the universe is fixed, $\inst I^i$
   is a finite instance, even though some constraints in it can be defined via infinitely many 
$n$-tree patterns.   It is easy to see that all the relations in the
    constraints are preserved by all the polymorphisms of $\Gamma$.
    
    The instance $\inst I^i$ is arc-consistent with sets $D_x^i(=D_x^{i+1})$:
    Let $((x_1,\dotsc,x_k),R)$ be a constraint defined by  $v_1,\dotsc, v_k$ in
    $t$ and let $a\in D_{x_j}^i$.  By \IPQ{n} there is a realization of $t$ in
    $D_x^{i+1}$ mapping $v_j$ to $a$ and thus  $D_{x_j}^i\subseteq \pr_j R$.  On
    the other hand, as $D_{x_j}^i=D_{x_j}^{i+1}$ and every tuple in $R$ comes
    from a realization inside the sets $D_x^{i+1}$'s, we get $\pr_j R\subseteq D_{x_j}^i$.
    
    Next we show that $\inst I^i$ has property (PQ).
    Part \ref{cond:pqac} of the definition was established in the paragraph
    above. For part \ref{cond:pqpq}, let $p$ and $q$ be arbitrary path patterns
    from $x$ to $x$ in $\inst I^i$.  
    Define $p'$ and $q'$ to be the paths of trees
    in $\inst I$ obtained, from $p$ and $q$, respectively, by replacing (in the natural way) each constraint in
    $p$ and $q$ with the~tree that defines it (we use the fact that each
    constraint is defined by leaves of a  tree).  We apply property \IPQ{n} for
    $\inst I$ with $\ell=i$ and patterns $p'$ and $q'$ to get that, for any
    $x\in V^i$ and any $a\in D_x^i$, there is a number $j$ such that $a\in
    \{a\}+^i(j(p'+q')+p')$.     The property (PQ) follows immediately.

    Since $\inst I^i$ has the property (PQ) then, by  Theorem~\ref{thm:nice}, it
    has a solution.  The solution to $\inst I$ will be obtained by taking the
    union of appropriately chosen solutions to $\inst I^1,\dotsc, \inst
    I^{|D|}$.

  \subsection{Invariant of the iterative construction}
    A global solution, denoted $\sol\colon V\to D$, 
    is constructed in steps. 
    At the start, we define it for the variables in $V^1$ by
    choosing an arbitrary solution to $\inst I^1$. 
    
    In step $i$ we extend the definition of $\sol$ from $V^{<i}$ to $V^{\le i}$,
    using a carefully chosen solution to $\inst I^i$.
    Our construction will maintain the following condition:
    \begin{description}
      \item{($E_i$)} every $n$-tree pattern in $\inst I$ has a realization inside the sets $D_x^{i+1}$
        which agrees with $\sol$ on $V^{\le i}$.
    \end{description}
    Note that, after the first step,  the condition $(E_1)$ is guaranteed 
    by the constraints of $\inst I^1$. 

    Assume that we are in step $i$: we have already defined $\sol$ on $V^{<i}$
    and condition $(E_{i-1})$ holds. Our goal is to extend $\sol$ by a~solution
    of $\inst I^i$ in such a~way that $(E_{i})$ holds.
    The remainder of Section~\ref{sect:proofof3} is devoted to proving that such
    a solution exists.

    Once we accomplish that, we are done with the proof:
    Condition $(E_{i})$ implies that $\sol$ is defined on $V^{\le i}$, and for
    every constraint $((x,y),R)$ between $x,y\in V^{\le i}$ the pattern from $x$
    to $y$ containing a~single edge labeled by $((x,y),R)$ is an~$n$-tree.  This
    implies that $\sol$ satisfies $((x,y),R)$ i.e., it is a solution on $V^{\le
    i}$.  After establishing $(E_{|D|})$ we obtain a solution to $\inst I$.

  \subsection{Restricting \texorpdfstring{$\inst I^i$}{the instance in level i}} 
    We begin by defining a new instance $\inst K^i$: it is defined almost
    identically to $\inst I^i$, but in part~\ref{enum:newconstraints} of the
    definition we require that the realization $r$ sends vertices from $V^{<i}$
    according to $\sol$.  As in the case of $\inst I^i$ we can assume that all
    the constraints are defined by leaves of the tree.  Thus every $n$-tree
    pattern with no internal vertices in $V^{i}$ defines one constraint in $\inst
    I^i$ and another in $\inst K^i$.
    Just like $\inst I^i$, the instance $\inst K^i$  is finite. 

    Note that we yet need to establish that constraints of $\inst K^i$ are
    non-empty, but the following claim, where $f$ is the fixed $(n+1)$-ary near
    unanimity polymorphism, holds independently.
    
    \begin{claim}\label{claim:absorbing}
      Let $((x_1,\dotsc,x_k),R)$ and $((x_1,\dotsc,x_k),R')$ be constraints 
      defined by the same tree $t$ in $\inst I^i$ and $\inst K^i$~(respectively).
      If $\tuple a_1,\dotsc,\tuple a_{n+1}\in R'$, $\tuple a\in R$, and $j\in
      \{1,\dots,n+1 \}$
      then $f(\tuple a_1,\dotsc,\tuple a_{j-1},\tuple a,\tuple a_{j+1},\dotsc,\tuple a_{n+1})$~
       belongs to $R'$.
    \end{claim}
    \begin{proof}
      Let $r_i$ be a realization of $t$ defining $\tuple a^i$; this realization
      sends all the vertices of $t$ labeled by variables from $V^{<i}$ according
      to $\sol$.  Let $r$ be a realization of $t$ defining $\tuple a$.

      Define a~function, from vertices of $t$ into $D$,
      sending a~vertex $v$ to
      \[
        f(r_1(v),\dotsc,r(v),\dotsc,r_{n+1}(v))
      \]
      (where $r(v)$ is in position $j$). 
      This is clearly a realization, and if $v$ is labeled by $x\in V^{<i}$ it sends $v$ according to $\sol$~
      (since $f$ is a near-unanimity operation).
      The new realization witnesses that $f(\tuple a_1,\dotsc,\tuple a_{j-1},\tuple a,\tuple a_{j+1},\dotsc,\tuple a_{n+1})$ belongs to $R'$.
    \end{proof}
    In order to proceed we need to show that the instance $\inst K^i$ contains a non-empty, arc-consistent subinstance, i.e., an arc-consistent instance (in some non-empty sets $D_x$) obtained from $\inst K^i$ by restricting every constraint in it so that each coordinate can take value only in the appropriate set $D_x$.
    
    A proof of this claim is the subject of the next section.

  \subsection{Arc-consistent subinstance of~$\inst K^i$}
    
    In order to proceed with the proof we need an additional definition.  Let
    $e\colon \inst J_1 \to \inst J_2$ be an instance homomorphism.  If for any
    variable $y$ of $\inst J_1$ and any constraint $((x_1,\dots,x_k),R)$ of
    $\inst J_2$ with $e(y) = x_i$~(for some $i$) the constraint
    $((x_1,\dots,x_k),R)$ has exactly one preimage $((y_1,\dots,y_k),R)$ with
    $y = y_i$, we say that $e$ is a~\emph{covering}.  A~\emph{universal
    covering tree instance} $\uct(\inst J)$ of a connected instance $\inst J$
    is a~(possibly countably infinite) tree instance $\inst T$ together with
    a~covering $e\colon\inst T\to \inst J$ satisfying some additional
    properties.  If $\inst J$ is a tree instance, then one can take $\uct(\inst
    J)=\inst J$, otherwise $\uct(\inst J)$ is always infinite.  If an instance
    $\inst J$ is disconnected  then $\uct(\inst J)$ is a disjoint union of
    universal covering tree instances for connected components of $\inst J$.
    
    Several equivalent (precise) definitions of $\UCT$ can be found in Section
    5.4 of~\cite{Kozik16:circles} or Section 4 of~\cite{Kun12:robust}. For our
    purposes, it is enough to mention that, for any $\inst J$, the instance
    $\uct(\inst J)$~(with covering $e$)  has the following two properties.  For
    any two variables $v,v'$ satisfying  $e(v)=e(v')$ there exists an
    endomorphism $h$ of $\uct(\inst J)$ (i.e., a~homomorphism into itself)
    sending $v$ to $v'$ and such that $e\circ h =e$.  Similarly for constraints
    $C$ and $C'$ if $e(C)=e(C')$ then there is an endomorphism $h$ such that
    $h(C)=C'$ and $e\circ h = e$.  It is well known that $\UCT(\inst J)$ has a
    solution if and only if $\inst J$ has an arc-consistent subinstance.

    Consider $\UCT(\inst K^i)$ and fix a~covering $e'\colon \UCT(\inst K^i)
    \rightarrow \inst K^i$. Let $\tuct^i$ be an instance obtained from
    $\UCT(\inst K^i)$ by replacing each constraint $C$ in it by a~tree that
    defines $e'(C)$, each time introducing a~fresh set of variables for the
    internal vertices of the trees.  Let $e$ be the instance homomorphism from
    $\tuct^i$ to $\inst I$ defined in the natural way.  We call a solution~(or
    a partial solution) to $\tuct^i$ {\em nice} if it maps each $v$ into
    $D_{e(v)}^{i+1}$ and moreover if $e(v)\in V^{<i}$ then $v$ is mapped to
    $\sol(e(v))$.  It should be clear that nice solutions to $\tuct^i$
    correspond to solutions of $\UCT(\inst K^i)$ (although the correspondence
    is not one-to-one).

    \begin{claim}\label{claim:almostAC}
      There exists a~nice solution of $\tuct^i$.
    \end{claim}
    \begin{proof}
      If $\tuct^i$ is not connected, we consider each connected component separately
      and then take the union of nice solutions.  Henceforth we assume that $\tuct^i$
      is connected.  By a standard compactness argument, it suffices to find a
      nice solution for every finite subtree of $\tuct^i$.  Suppose,
      for a contradiction, that $\inst T$ is a minimal finite subtree of $\tuct^i$
      without nice solutions.
      
      First, only the leaf vertices of $\inst T$ can be mapped, by $e$, into
      variables from $V^{<i}$.  Indeed, if an internal vertex is mapped to a variable in $V^{<i}$, we can
      split the tree at this vertex into two parts, obtain~(from the minimality of $\inst T$)
      nice solutions to both parts~(which need to map the splitting vertex
      according to $\sol$, i.e., to the same element) and merge these solutions
      to obtain a nice solution to $\inst T$.  This is a contradiction.

Second we show that $\inst T$ has more than $n$ leaves mapped by $e$ into
$V^{<i}$.  Assume that $\inst T$ has $n$ or fewer leaves mapped to $V^{<i}$ and
let $\inst T'$ be the smallest subtree of $\inst T$ with these leaves.
Then $\inst T'$ is an $n$-tree and by $(E_{i-1})$ we obtain a solution $s$ to
$\inst T'$ in $D^i_x$'s which sends leaves of $\inst T'$ according to $\sol$.
It remains to extend $s$ to a solution of $\inst T$ in $D^{i+1}_x$'s. This
extension is done in a sequence of steps. In each step $s$ is defined for
increasingly larger subtrees of $\inst T$. Furthermore, in each step the
following condition (*) is satisfied by $s$:
if a vertex $v$ has a value assigned by $s$ and a neighbour without such value
then $s(v)$ belongs to $D^i_{e(v)}$.
Clearly, this condition holds in the beginning.
In each step we pick a~constraint $C$ on a~vertex $v$ with an assigned value and
a~vertex $v'$ without such a value. (Note that the constraints of $\inst T^i$,
and consequently of $\inst T$, are binary.)
$C$ has been added to $\inst T^i$ by replacing a constraint of $\uct(\inst K^i)$ with an $n$-tree ${\inst T_C}$ that defines it.
Let $\inst S$ be a~maximal subtree of ${\inst T}$ such that
it contains $C$, it has $v$ as a leaf, and all other nodes in $\inst S$ have not been assigned by $s$ and belong
to ${\inst T_C}$. 
Since ${\inst T_C}$ is a $n$-tree, $\inst S$ is also an $n$-tree, and we can use $\IPQ{n}$ to derive that there
exists a solution, $s'$, of ${\inst S}$ in $D_x^{i+1}$'s that sends $v$ to $s(v)\in D^i_{e(v)}$. More specifically,
we apply $\IPQ{n}$ with $x=v$, $a=s(v)$, and both $p$ and $q$ being the same pattern $t_1+t_2$ such that $t_1$ is  $\inst S$ with beginning $v$ and end being any other leaf of $\inst S$, and $t_2$ is $t_1$ with beginning and end swapped.
This solution $s'$ can be added to $s$ (as
the values on $v$ are the same). It remains to see that condition (*) is preserved after extending $s$ with $s'$. Indeed, let $u$ be
any vertex such that after adding solution $s'$ has a neighbour $u'$ that has not yet been assigned. We can assume that $u$
is one of the new variables assigned by $s'$. If $e(u)\in V^i$ then the claim
follows from the fact that $D^{i+1}_{e(u)}=D^i_{e(u)}$ so we can assume that
$e(u)\not\in V^i$. However, in this case, all neighbours of $u$ in $\inst T$ must be in $\inst T_C$, so the constraint in ${\mathcal T}$
containing both $u$ and $u'$ must be also in ${\inst T_C}$ contradicting the
maximality of ${\inst S}$.

      So the counterexample $\inst T$ must have at least $n+1$ leaves mapped
      into $V^{<i}$.  Fix any $n+1$ of such leaves $v_1,\dotsc,v_{n+1}$ and let
      $\inst T_j$, for $j=1,\dots,n+1$, denote a~subinstance of $\inst T$
      obtained by removing $v_j$ together with the single constraint containing
      $v_j$: $((v_j,v'_j),R_j)$ from $\inst T$.  Clearly, $v'_j$ is not a
      leaf~(as it would make our $\inst T$ a two-element instance) and by the
      fact that only leaves can be mapped into $V^{<i}$ we get that $e(v'_j)
      \in V^{i}$ or $e(v'_j)\in V^{>i}$ and, in the last case, $i\neq|D|$.
      
      By minimality, each $\inst T_j$ has a~nice realization, say $s_j$.  Now
      either $e(v'_j)\in V^i$ and $s_j(v'_j)\in D_{e(v'_j)}^i =
      D_{e(v'_j)}^{i+1}$ or $e(v'_j)\in V^{>i}, s_i(v'_j)\in D_{e(v'_j)}^{i+1}$
      and $i+1\neq |D|+1$.  In both cases $s_j(v'_j)\in D_{e(v'_j)}^{i'}$ for
      $i'\leq |D|$ and thus, by \IPQ{n}, there exists $a_j\in D$ such that
      $(a_j,s_j(v'_j))\in R_i$.  We let $s'_j$ be the realization of $\inst T$
      obtained by extending $s_j$ by mapping $v_j$ to $a_j$.  The last step is
      to apply the $(n+1)$-ary near unanimity operation coordinatewise to
      $s'_j$'s (in a~way identical to the one in the proof of
      Claim~\ref{claim:absorbing}).  The application produces a~nice
      realization of $\inst T$.  This contradiction finishes the proof of the
      claim.
    \end{proof}
    
    We will denote the arc-consistent subinstance of $\inst K^i$~
    (which is about to be constructed) by $\inst L^i$.
    The variables of $\inst L^i$ and $\inst K^i$~(or indeed $\inst I^i$) are the same.
    For every constraint $(\tuple x, R)$ in $\inst K^i$
    we introduce a constraint $(\tuple x, R')$ into $\inst L^i$
    where 
    \begin{equation*}
      R' = \{\tuple{a}\colon \tuple a = s(\tuple y) \text{ where $s$ is a solution to $\UCT(\inst K^i)$ and $e'((\tuple y,R)) = (\tuple x, R)$}\}
    \end{equation*}
    where $e'$ is an instance homomorphism mapping $\UCT(\inst K^i)$ to $\inst K^i$.
    In other words we restrict a relation in a constraint of $\inst K^i$ by allowing only the tuples 
    which appear in a solution of the $\UCT(\inst K^i)$~
    (at this constraint).
    
    All the relations of $\inst L^i$ are preserved by all the polymorphisms of $\Gamma$,
    and are non-empty~(by Claim 2).
    The fact that $\inst L^i$ is arc-consistent
    is an easy consequence of the endomorphism structure of universal covering trees.
    Finally Claim~\ref{claim:absorbing} holds for $\inst L^i$:

    \begin{claim}\label{claim:absorbingL}
      Let $((x_1,\dotsc,x_k),R)$ and $((x_1,\dotsc,x_k),R')$ be constraints 
      defined by the same tree $t$ in $\inst I^i$ and $\inst L^i$, respectively.
      Let $\tuple a_1,\dotsc,\tuple a_{n+1}\in R'$ and  $\tuple a\in R$, 
      then $f(\tuple a_1,\dotsc,\tuple a,\dotsc,\tuple a_{n+1})$, 
      where $f$ is the $(n+1)$-ary near unanimity operation and  
      $\tuple a$ is in position $j$,  belongs to $R'$.
    \end{claim}
    \begin{proof}
      By Claim~\ref{claim:absorbing} the tuple $f(\tuple a_1,\dotsc,\tuple a,\dotsc,\tuple a_{n+1})$
      belongs to the relation in the corresponding constraint in $\inst K^i$. 
      Thus if it extends to a solution of $\UCT(\inst K^i)$ it belongs to $R'$.
      However each $\tuple a^i$ extends to a solution of $\UCT(\inst K^i)$ 
      and $\tuple a$ extends to a solution of $\UCT(\inst I^i)$. 
      By applying the near-unanimity operation $f$ to these extensions~(coordinatewise),
      we obtain the required evaluation.
    \end{proof}

  \subsection{A solution to $\inst K^i$}

    In order to find a solution to $\inst L^i$, we will use Corollary B.2
    from~\cite{Kozik16:circles}.  
    We state it here in a~simplified form using
    the following notation: for subuniverses $A'\subseteq A$, we say that $A'$
    {\em nu-absorbs} $A$ if, for some NU polymorphism $f$, $f(a_1,\ldots,a_n)\in
    A'$ whenever $a_1,\ldots,a_n\in A$ and at most one $a_i$ is in $A\setminus A'$.
    Similarly, if $R'\subseteq R$ are relations preserved by all
    polymorphisms of $\Gamma$ we say $R'$ nu-absorbs $R$, if for some
    near-unanimity operation $f$ taking all arguments from $R'$ except for one
    which comes from $R$ produces a result in $R'$.
    
    \begin{corollary}[Corollary B.2 from~\cite{Kozik16:circles}]
      Let $\inst I$ satisfy (PQ) condition in sets $A_x$.
      Let $\inst I'$ be an arc-consistent instance in sets $A'_x$ on the same set of
      variables as $\inst I$ such that:
      \begin{enumerate}
        \item for every variable $x$ the subuniverse $A'_x$ nu-absorbs $A_x$, and
        \item for every constraint $((x_1,\dotsc,x_n),R')$ in $\inst I'$ there
          is a corresponding constraint $((x_1,\dotsc,x_n),R)$ in $\inst I$ such
          that $R'$ nu-absorbs $R$~(and both respect the NU operation).
      \end{enumerate}
      Then there are subuniverses $A_x''$ of $A_x'$~(for every $x$) such that
      the instance $\inst I''$ obtained from $\inst I'$ by restricting the
      domain of each variable to $A''_x$ and by restricting the constraint
      relations accordingly satisfies the condition (PQ).
    \end{corollary}

    We will apply the corollary above using $\inst I^i$ for $\inst I$ and $\inst L^i$ for $\inst I'$.
    By our construction, $\inst I^i$ satisfies condition (PQ), and the sets $D_x^i$~
    (which play the role of $A_x$) 
    are subuniverses of $D$.
    On the other hand $\inst L^i$ is arc-consistent and 
    all the relations involved in it are closed under the polymorphisms of $\Gamma$.
    Claim~\ref{claim:absorbingL} shows that each relation $R'$ nu-absorbs the corresponding $R$.
By arc-consistency, the projection of $R'$ on a variable $x$ is the same for each constraint $((x_1,\ldots,x_n),R')$ containing $x$, call the corresponding sets $A'_x$. Since each $R'$
 nu-absorbs $R$, it follows that each $A'_x$ nu-absorbs the corresponding $A_x$.  
    The corollary implies that we can restrict 
    the instance $\inst L^i$ to obtain
    an instance satisfying (PQ).
    By Theorem~\ref{thm:nice} such an instance, and thus both $\inst K^i$ and $\inst L^i$, has a solution.
   
  \subsection{Finishing the proof}
    We choose any solution to $\inst K^i$ and extend the global solution $\sol$
    to $V^i$ according to it. 
    There exists a solution on $V^{\le i}$, 
    because every constraint between two variables from this set is either in $V^{<i}$ or
    defines a two-variable $n$-tree which was used to define a constraint in $\inst K^i$. 
    It remains to prove
    that, with such an extension, condition $(E_{i})$ holds.
    
    Let $t$ be an $n$-tree pattern in $\calI$.  If it has no variables mapped to
    $V^{i}$, then $(E_i)$ follows from $(E_{i-1})$. Assume that it has such variables. By splitting $t$ at internal vertices mapped to $V^i$, it is enough to consider the case when only leaves of $t$ are mapped to $V^i$. Then  
    $t$ defines a constraint $(\tuple x,R)$ in $\inst K^i$.  The solution to
    $\inst K^i$ mapping $\tuple x$ to $\tuple a\in R$ and the evaluation of $t$
    witnessing that $\tuple a$ belongs to $R$ can be taken to satisfy $(E_i)$
    for $t$.
    Theorem~\ref{thm:nicelevels} is proved.

\section{Full proof of Theorem \ref{the:main}(1)}
      \label{sec:correctness-of-algorithm}
      
      In this subsection we prove Propositions \ref{prop:expected-lost-weight}
      and~\ref{prop:consistence}.
      The following equalities, which can be directly verified, are used
      repeatedly in this section: for any subsets $A,B$ of $D$ and any feasible
      solution $\{\bx_a\}$ of the SDP relaxation of $\inst I$ it holds that
      $\|\bx_A\|^2=\bx_A\by_D$ and $\|\by_B-\bx_A\|^2=\bx_{D\setminus
      A}\by_B+\bx_A\by_{D\setminus B}$. 

      \subsection{Analysis of Preprocessing step 2}\label{sec:preproc}

        In some of the proofs it will be required that $\alpha\leq c_0$ for some
        constant $c_0$ depending only on $|D|$.  This can be assumed without
        loss of generality, since we can adjust constants in $O$-notation in
        Theorem~\ref{the:main}(1) to ensure that $\eps\le c_0$ (and we know that
        $\alpha\le \eps$). We will specify the requirements on the choice of
        $c_0$ as we go along.

        \begin{lemma} \label{le:defH} \label{lem:maj2}
        There exists a constant $c>0$ that depends only on $|D|$ such that the
        sets $D_x^{\ell}\subseteq D$, $x\in V$, $1\leq \ell\leq |D|$, obtained
        in Preprocessing step 2, are non-empty and satisfy the following
        conditions:
        \begin{enumerate}
          \item for every $a\in D_x^{\ell}$, $\|\bx_a\|\geq \alpha^{3\ell \kappa}$,
          \item for every $a\not\in D_x^{\ell}$,
            $\|\bx_a\|\leq c\alpha^{3\ell \kappa}$,
          \item for every $a\in D_x^{\ell}$,
            $\|\bx_a\|^2\geq 2 \|\bx_{D\setminus{D_x^{\ell}}}\|^2$,
          \item $D_x^{\ell}\subseteq D_x^{\ell+1}$ (with $D_x^{|D|+1}=D$).
        \end{enumerate}
        \end{lemma}

        \begin{proof}
          Let $c=(2|D|)^{(|D|/2)}$. It is straightforward to verify that
          conditions (1)--(3) are satisfied.  Let us show condition (4). Since
          $c$ only depends on $|D|$ we can choose $c_0$ (an upper bound on
          $\alpha$) so that $c\alpha^{3\kappa}<1$. It follows that
          $c\alpha^{3(\ell+1)\kappa}<\alpha^{3\ell \kappa}$. It follows from
          conditions (1) and (2) that $D_x^{\ell}\subseteq D_x^{\ell+1}$.

          Finally, let us show that $D_x^{\ell}$ is non-empty. By condition (4)
          we only need to take care of case $\ell=1$. We have by condition (2)
          that
          \[
            \sum_{a\in {D\setminus D_x^1}} \|\bx_a\|^2\leq |D|c^2\alpha^{6\kappa}
          \]
          Note that we can adjust $c_0$ to also satisfy
          $|D|c^2\alpha^{6\kappa}<1$ because, again, $c$ only depends on $|D|$.
        \end{proof}

      \subsection{Proof of Proposition \ref{prop:expected-lost-weight}}
         
         We will prove that the total weight of constraints removed in each step 0-5 of the algorithm in Section~\ref{sec:algorithm} is $O(\alpha^\kappa)$.

        \begin{lemma} \label{le:step0}
          The total weight of the constraints removed in step $0$ is at most
          $\alpha^{\kappa}$.
        \end{lemma}

        \begin{proof}
          We have
          \[
            \alpha\geq \sum_{C\in{\mathcal C}} w_C \loss(C)\geq
            \sum_{\substack{C\in{\mathcal C} \\ \loss(C)\geq\alpha^{1-\kappa}}} w_C
            \alpha^{1-\kappa},
          \]
          from which the lemma follows.
        \end{proof}

        \begin{lemma} \label{le:path-loss}
          Let $((x,y),R)$ be a constraint not removed in step $0$, and let $A,B$
          be such that $B=A+^{\ell}(x,R,y)$. Then
          $\|\by_B\|^2\geq \|\bx_A\|^2-c\alpha^{(6\ell+6)\kappa}$
          for some constant $c>0$ depending only on $|D|$. The same is also true
          for a~constraint $((y,x),R)$ and $A = B+^\ell (y,R^{-1},x)$.
        \end{lemma}
        \begin{proof}
        Consider the first case, i.e., a~constraint $((x,y),R)$ and $B = A+^\ell
        (x,R,y)$.
        We have
        \[
          \bx_A\by_{D\setminus B} =
            \sum_{\substack{a\in A, b\in D\setminus B \\ (a,b)\not\in R}} \bx_a\by_b +
            \sum_{\substack{a\in A, b\in D\setminus B \\ (a,b)\in R}} \bx_a\by_b.
        \]
        The first term is bounded from above by the loss of constraint
        $((x,y),R)$, and hence is at most $\alpha^{1-\kappa}$, since the
        constraint has not been removed in step $0$.
        Since $B=A+^{\ell}(x,R,y)$ it follows that for every $(a,b)\in R$ such
        that $a\in A$ and $b\in D\setminus{B}$ we have that $a\not\in
        D_x^{\ell+1}$ or $b\not\in D_y^{\ell+1}$. Hence, the second term is at
        most
        \[
        \bx_{D\setminus D_x^{\ell+1}}\by_D+\bx_D\by_{D\setminus
        D_y^{\ell+1}}=\|\bx_{D\setminus D_x^{\ell+1}}\|^2+\|\by_{D\setminus
        D_y^{\ell+1}}\|^2
        \]
        which, by Lemma~\ref{le:defH}(2), is bounded from
        above by $d\alpha^{(6\ell+6)\kappa}$ for some constant $d>0$. From the
        definition of $\kappa$ it follows that $(6\ell+6)\kappa\leq 1-\kappa$,
        and hence we conclude that $\bx_A\by_{D\setminus B}\leq
        (d+1)\alpha^{(6\ell+6)\kappa}$.
        Then, we have that
        \begin{multline*}
          \|\by_B\|^2 = \bx_A\by_B+\bx_{D\setminus A}\by_B
              \geq \bx_A\by_B = \\
            \bx_A\by_D-\bx_A\by_{D\setminus B}
              \geq \|\bx_A\|^2-(d+1)\alpha^{(6\ell+6)\kappa}.
              \quad\qedhere\!\!\!\!\!\!
        \end{multline*}
        \end{proof}

        \begin{lemma} \label{le:step1}
          The expected weight of the constraints removed in step $1$ is
          $O(\alpha^{\kappa})$.
        \end{lemma}

        \begin{proof}
          Let $((x,y),R)$ be a constraint not removed in step $0$. We shall see
          that the probability that it is removed in step $1$ is at most $c
          \alpha^{\kappa}$ where $c>0$ is a constant.

          Let $A,B$ be such that $B=A+^{\ell} (x,R,y)$. It follows from Lemma
          \ref{le:path-loss} that $\|\by_B\|^2\geq
          \|\bx_A\|^2-d\alpha^{(6\ell+6)\kappa}$ for some constant $d>0$.
          Hence, the probability that a value $r_\ell$ in step $1$ makes that
          $\by_B\not\preceq^{\ell} \bx_A$ is at most
          \[
            \frac{d\alpha^{(6\ell+6)\kappa}}{\alpha^{(6\ell+4)\kappa}}
              = d\alpha^{2\kappa}\leq d \alpha^{\kappa}.
          \]
          We obtain the same bound if we switch $x$ and $y$, and consider sets
          $A,B$ such that $A=B+^{\ell} R^{-1}$. Taking the union bound for all
          sets $A,B$ and all values of $\ell$ we obtain the desired bound.
        \end{proof}

        \begin{lemma} \label{lem:removing-variable}
          If there exists a~constant $c>0$ depending only on $|D|$ such that for
          every variable $x$, the probability that all constraints involving $x$ are
          removed in step~2, step~3, or step~5 is at most $c\alpha^\kappa$, then
          the total expected weight of constraints removed this way in the corresponding is at most $2c\alpha^\kappa$.
        \end{lemma}
        \begin{proof}
          Let $w_x$ denote the total weight of the constraints in which $x$
          participates. The expected weight of constraints removed is at most
          \[
            \sum_{x\in V} w_x c \alpha^{\kappa}
              = (\sum_{x\in V} w_x) c\alpha^{\kappa} = 2 c\alpha^{\kappa} 
          \]
          and the lemma is proved.
        \end{proof}

        \begin{lemma}
          The expected weight of the constraints removed in step 2 is
          $O(\alpha^\kappa)$.
        \end{lemma}

        \begin{proof}
          Let $x$ be a~variable. According to Lemma~\ref{lem:removing-variable} it is 
          enough to prove that the probability that we
          remove all constraints involving $x$ in step 2 is at most
          $c\alpha^\kappa$ for some constant $c>0$.  
          Suppose that $A\subseteq B$ are such that $\| \vec x_B \|^2 - \| \vec
          x_A \|^2 = \|\vec x_B - \vec x_A \|^2
          \leq (2n -3) \alpha^{(6\ell+4)\kappa}$. Then the probability that
          one of the bounds of the form $r_\ell + (s_\ell +jm_0)\alpha^{(6\ell
          +4)\kappa}$ separates $\| \vec x_B \|^2$ and $\| \vec x_A \|^2$ is at
          most
          \[
            {(2n-3)}/{m_0} \leq 
            (2n-3)/ (\alpha ^{-2\kappa} - 1) 
          \]
          which is at most $c\alpha^\kappa$ for some constant $c>0$ whenever
          $\alpha^\kappa < 1/2$. The latter can be ensured by adjusting constant $c_0$ from
          Section~\ref{sec:preproc}.
          Taking the union bound for all sets $A,B$ and all values of $\ell$ we
          obtain the desired bound.
          \end{proof}

        \begin{lemma} \label{le:cut}
        There exist constants $c,d>0$ depending only on $|D|$ such that for
        every pair of variables $x$ and $y$ and every $A,B\subseteq D$, the
        probability, $p$, that a unit vector $\bu$ chosen uniformly at random
        cuts $\bx_A$ and $\by_B$ satisfies
        \[
          c\cdot \| \by_B-\bx_A \| \leq p\leq d\cdot \|\by_B-\bx_A \|.
        \]
        \end{lemma}

        \begin{proof}
          Let $0\leq x\leq 1$ and let $0\le \theta\le \pi$ be an angle such that
          $x=\cos(\theta)$. There exist constants $a,b>0$ such that  $$a\cdot
          \sqrt{1-x}\leq \theta\leq b\cdot \sqrt{1-x}.$$
          Now, if $\theta$ is the angle between $\bx_A-\bx_{D\setminus A}$ and
          $\by_B-\by_{D\setminus B}$ then
          \begin{multline*}
            1-\cos(\theta) =
              1-(\bx_A-\bx_{D\setminus A})(\by_B-\by_{D\setminus B}) = \\
              2(\bx_{D\setminus A}\by_B+\bx_A\by_{D\setminus B}) =
                2\|\by_B-\bx_A \|^2
          \end{multline*}
          Since $p=\theta/\pi$, the result follows.
        \end{proof}

        \begin{lemma} \label{le:maj-step3} \label{le:step4}
        The expected weight of the constraints removed in step 3 is
        $O(\alpha^{\kappa})$.
        \end{lemma}
        \begin{proof}
          According to Lemma \ref{lem:removing-variable}, it is enough to prove
          that the probability that we remove all constraints involving $x$ in
          step 3 is at most $c\alpha^{\kappa}$ for some constant $c$.
          Let $A$ and $B$ be such that $A\cap D_x^{\ell}\neq B\cap D_x^{\ell}$. Let
          $a$ be an element in  symmetric difference $(A\cap
          D_x^{\ell})\triangle (B\cap D_x^{\ell})$. Then we have $\|\bx_B-\bx_A
          \|=\sqrt{\bx_{D\setminus A}\bx_B+\bx_A\bx_{D\setminus B}}\geq
          \|\bx_a\|\geq \alpha^{3\ell \kappa}$, where the last inequality is by
          Lemma~\ref{le:defH}(1).  Then by Lemma \ref{le:cut} the probability
          that $\bx_A$ and $\bx_B$ are not $\ell$-cut is at most
          \[
            (1-c\alpha^{3\ell\kappa})^{m_{\ell}} \leq
              \frac{1}{\exp(c\alpha^{3\ell\kappa}m_{\ell})} \leq 
              \frac{1}{\exp(c\alpha^{-\kappa})} \leq c\alpha^{\kappa}.
          \]
          where $c$ is the constant given in Lemma \ref{le:cut}.
       Taking the union bound for all sets $A,B$ and all values of $\ell$ we
          obtain the desired bound.
        \end{proof}

        \begin{lemma} \label{le:maj-step2} \label{le:step3}
          The expected weight of the constraints removed in step~4 is
          $O(\alpha^{\kappa})$.
        \end{lemma}
        \begin{proof}
          Let $((x,y),R)$ be a constraint not removed in steps $0$ and $1$. We shall
          prove that the probability that it is removed in step 4 is at most
          $c\alpha^{\kappa}$ for some constant $c>0$.

          Fix $\ell$ and $A,B$ such that $B=A+^{\ell}(x,R,y)$. Since the constraint has not been removed in step 1, we   have  $\by_B\preceq^{\ell} \bx_A$.  Since $B=A+^{\ell} p$ we have that
          $\bx_A \by_{D\setminus B}\leq c_1\alpha^{(6\ell+6)\kappa}$, as shown in
          the proof of Lemma~\ref{le:path-loss}. Since $\|\bx_A\|^2=\bx_A(\by_B+
          \by_{D\setminus B})$, it follows that  $\bx_A \by_B\geq
          \|\bx_A\|^2-c_1\alpha^{(6\ell+6)\kappa}$.

          Also, we have $\|\by_B\|^2=(\bx_A\by_B+\bx_{D\setminus A}\by_B)$ is at
          most $\|\bx_A\|^2+\alpha^{(6\ell+4)\kappa}$ because
          $\by_B\preceq^{\ell} \bx_A$.  Using the bound on  $\bx_A \by_B$
          obtained above, it follows that $\bx_{D\setminus A}\by_B$ is at most
          $\alpha^{(6\ell+4)\kappa}+c_1\alpha^{(6\ell+6)\kappa}\leq
          (c_1+1)\alpha^{(6\ell+4)\kappa}$.

          Putting the bounds together, we have that
          \begin{multline*}
            \|\by_B-\bx_A \| = 
              \sqrt{\bx_{D\setminus A}\by_B + \bx_A\by_{D\setminus B}} \leq \\
              \sqrt{c_1\alpha^{(6\ell+6)\kappa} + (c_1+1)\alpha^{(6\ell+4)\kappa}}
              \leq c_2 \alpha^{(3\ell+2)\kappa}
          \end{multline*}
          for some constant $c_2>0$.

          Applying the union bound and Lemma \ref{le:cut} we have that the
          probability that $\bx_A$ and $\by_B$ are $\ell$-cut is at most
          $m_{\ell} dc_2\alpha^{(3\ell+2)\kappa}=O(\alpha^{\kappa})$. We obtain
          the same bound if we switch $x$ and $y$, and take $R^{-1}$ instead
          of~$R$.
          Taking the union bound for all sets $A,B$ and all values of $\ell$ we
          obtain the desired bound.
        \end{proof}

        \begin{lemma}
          The expected weight of the constraints removed in  step 5 is
          $O(\alpha^\kappa)$.
        \end{lemma}

        \begin{proof}
          Again, according to Lemma \ref{lem:removing-variable}, it is enough to
          prove that the probability that we remove all constraints involving
          $x$ in step~5 is at most $c_1\alpha^{\kappa}$ for some constant $c_1$.
          Suppose that $A$, $B$ are such that $\| \vec x_A - \vec x_B \|^2 \leq
          (2n-3)\alpha^{(6\ell+4)\kappa}$. Hence, by Lemma \ref{le:cut} and the union 
          bound the probability that $\vec x_A$ and $\vec x_B$ are $\ell$-cut is at
          most
          \[
            m_\ell d(2n-3)^{1/2}\alpha^{(3\ell+2)\kappa}
              \leq d(2n-3)^{1/2}\alpha^ \kappa
          \]
          where $d$ is the constant from Lemma \ref{le:cut}.
          Taking the union bound for all sets $A$, $B$ and all values of $\ell$,
          we obtain the desired bound.
        \end{proof}

  \subsection{Proof of Proposition \ref{prop:consistence}}

  All patterns appearing in this subsection are in $\calI'$.  The following
  notion will be used several times in our proofs: Let $t$ be a tree and let
  $y$ be one of its nodes. We say that a subtree $t'$ of $t$ is {\em separated
  by vertex $y$} if $t'$ is maximal among all the subtrees of $t$ that contain
  $y$ as a~leaf.

  In the first part of the proof (which consists of the following three
  lemmas), we prove that if we start with a~set $A \subseteq D_x$ and propagate it via
  a~path $p$, from $x$ to $y$, of $n$-tree patterns to obtain a set $B \subseteq D_y$, the value $\|\by_B \|$ cannot be much smaller than $\|\bx_A \|$. The first
  lemma proves that this is the case if we restrict to proprer path patterns.
  
  \begin{lemma} \label{le:notdecreasing} \label{lem:maj18}
  \label{lem:path-loss}
    Let $1\leq \ell\leq |D|$, let $p$ be a~path pattern from $x$ to $y$,  and
    let $A,B$ be such that $B=A+^{\ell} p$. Then $\bx_A\preceq^{\ell} \by_B$,
    and in particular,
    $\|\bx_A \| \leq \|\by_B \| + \alpha^{(6\ell+4)\kappa}$.
  \end{lemma}
  \begin{proof}
    Since the relation $\preceq^\ell$ is transitive, it is enough to
    prove the lemma for path patterns containing only one constraint.  But
    this is true, since all the constraints $((x,y),R)$ or $((y,x),R)$
    which would invalidate the lemma have been removed in step 1.
  \end{proof}

  The second lemma proves that the weight of sets that vanish after following a~tree pattern is small.

  \begin{lemma} \label{lem:lost-on-a-tree}
  If $p$ is a~tree pattern with at most $j+1$ leaves starting at $x$, and
  $A \subseteq D_x^{\ell+1}$ is such that $A +^\ell p = \emptyset$ then
  \( \| \vec x_A \|^2 \leq (2j-1) \alpha^{(6\ell  +4) \kappa} \).
  \end{lemma}

  \begin{proof} We will prove the statement by induction on the number of
  leaves. For $j=1$ this follows from Lemma \ref{lem:maj18}.
  Suppose then that $p$ is a~tree pattern with $j+1>2$ leaves and the
  statement is true for any tree pattern with at most $j$ leaves. Choose $y$ to be the first branching
  vertex in the unique path in $p$ from $x$ to the end of $p$,
  and let $p_0,t_1,\dots,t_h$ be all subtrees of $p$ separated by $y$ where $p_0$ 
  is the subtree containing $x$. We turn $p_0$ into a pattern by choosing $x$ as 
  beginning 
  and $y$ as end. Similarly, we turn every $t_i$ into a pattern by choosing $y$ 
  as beginning and any other arbitrary leaf as end.
  Since $y$ is
  a~branching vertex, we have that $h \geq 2$, every $t_i$ has $j_i + 1 <
  j + 1$ leaves, and
  $\sum_{i=1}^h{j_i} = j$.
  Now, let $B_i$ denote the set $\{ a\in D_y^{\ell +1} : \{a\} +^\ell t_i =
  \emptyset \}$.
  Since $j_i < j$, we know that $\| \vec y_{B_i} \|^2 \leq (2j_i -1)
  \alpha^{(6\ell + 4)\kappa}$. Further, for $B = \bigcup_{i=1}^h B_i$, we have, using inductive assumption, that
  \begin{multline*}
    \|\vec y_B\|^2 \leq \sum_{i=1}^h \|\vec y_{B_i}\|^2 \leq
      \sum_{i=1}^h (2j_i -1) \alpha^{(6\ell + 4)\kappa} \\
      = (2j - h) \alpha^{(6\ell + 4)\kappa}
      \leq (2j - 2) \alpha^{(6\ell +4)\kappa}.
  \end{multline*}
  Finally, since $A +^\ell p = \emptyset$ then $A +^\ell p_0 \subseteq B$,
  and the claim follows from Lemma~\ref{lem:maj18}.
  \end{proof}

  The following lemma concludes the first part of the proof by proving that following a~path of $n$-trees pattern cannot decrease the weight of a~set too much.

  \begin{lemma} \label{lem:tree-loss} \label{lem:14}
    Let $1\leq \ell \leq |D|$, let $p$ be a pattern from $x$ to $y$ 
    which is a~path of $n$-trees. If $A,B\subseteq D$ are such
    that $A +^\ell p = B$, then \( \|\vec x_A\|^2 \leq \|\vec y_B\|^2 +
    \alpha^{(6\ell +2)\kappa} \).
  \end{lemma}

  \begin{proof}
    We claim that for any $n$-tree pattern $t$ and $A,B$ with $A+^\ell t = B$,
    we have $\vec x_A \preceq_w^\ell \vec y_B$. Since the relation
    $\preceq_w^\ell$ is  transitive, the lemma is then a~direct consequence.
    For a~contradiction, suppose that $t$ is a~smallest (by inclusion) $n$-tree
    that does not satisfy the claim. Observe that $t$ is not a~path, due to
    Lemma~\ref{lem:path-loss} and the fact that $\vec x_A\preceq^\ell \vec y_B$
    implies $\vec x_A\preceq_w^\ell \vec y_B$. Let $v_x$ and $v_y$ denote the
    beginning and the end vertex of $t$, respectively; and let $v_z$ be the
    last branching vertex that appears on the path connecting $v_x$ and $v_y$,
    and let it be labeled by $z$.  Let $t_1,t_2,p_1,\dots,p_j$ be all subtrees
    of $t$ separated by $v_z$, where $t_1$ and $t_2$ are the subtrees
    containing $v_x$ and $v_y$ respectively.  Let us turn $p_1,\dots,p_j$ into
    patterns by choosing $v_z$ as beginning and any other leaf as end. Note
    that the sum of numbers of the leaves of $p_1,\dots,p_j$ when excluding
    $v_z$ is less than $n-1$ since $t$ was a~path of $n$-trees.  Furthermore,
    choose $x$ and $z$ to be the beginning and end, respectively, of $t_1$ and
    $z$ and $y$ to be the beginning and end, respectively, of $t_2$. Note that
    $t_2$ is a~path.  Further, we know that for $C = A +^\ell t_1$ we have
    $\vec x_A \preceq_w^\ell \vec z_C$ by minimality of $t$. Now, let $C_i = \{
    a \in D_z^{\ell +1} : \{a\}+^\ell p_i = \emptyset \}$. Then by Lemma
    \ref{lem:lost-on-a-tree}, we get that $\|\vec z_{C_i} \|^2 \leq (2j_i
    -1)\alpha^{(6\ell +4)\kappa}$ where $j_i+1$ is the number of leaves of
    $p_i$, therefore for $C' = \bigcup C_i$ we have $\| \vec z_{C'} \|^2 \leq
    \sum\| \vec z_{C_i} \|^2\le (2n - 3) \alpha^{(6\ell +4)\kappa}$ (we used
    that $\sum j_i \leq n-1$).  This implies that $\| \vec z_{C \setminus C'}
    \|^2 \geq \| \vec z_C \|^2 -  (2n - 3) \alpha^{(6\ell +4)\kappa}$, and
    consequently $\vec z_C \preceq_w^\ell \vec z_{C \setminus C'}$ as otherwise
    all constraints containing $z$ would have been removed in step 2.  Finally,
    observe that $B = (C\setminus C') +^\ell t_2$, and therefore $\vec z_{C
    \setminus C'} \preceq^\ell \vec y_B$ and, hence, $\vec z_{C \setminus C'}
    \preceq_w^\ell \vec y_B$. Putting this together with all other derived
    $\preceq^\ell_w$-relations, we get the required claim.
  \end{proof}

  Next, we move to proving the condition $\IPQ{n}$. For that we will need the following technical statement. Intuitively, the statment says that, starting with a~set $A$, if we follow a~circular path of $n$-tree patterns and end up back in the set $A$, then all values from $A$ can be reached by this pattern.

  \begin{lemma} \label{lem:maj19}
    Let $1 \leq \ell \leq |D|$, let $p$ be a~pattern from $x$ to $x$ which
    is a~path of $n$-trees, and let $A,B$ be such that $A +^\ell p = B$.
    If $B \cap D_x^\ell \subseteq A \cap D_x^\ell$ then $A \cap D_x^\ell =
    B \cap D_x^\ell$.
  \end{lemma}

  \begin{proof}
    For a~contradiction, suppose that there is an~element $a \in (D_x^\ell
    \cap A) \setminus B$. From Lemma~\ref{lem:maj2} we get that
    \(
      \| \vec x_{A\setminus B} \|^2 \geq \|\vec x_a \|^2
      \geq 2 \|\vec x_{D\setminus D_x^\ell} \|^2
      \geq 2 \|\vec x_{B\setminus A} \|^2.
    \)
    Therefore, we have
    \begin{multline*}
      \|\vec x_B\|^2
      = \|\vec x_A\|^2 - \|\vec x_{A\setminus B}\|^2
        + \|\vec x_{B\setminus A}\|^2 \leq
        \|\vec x_A \|^2 - \| \vec x_a \|^2 + (1/2)\| \vec x_a \|^2 \\
        = \|\vec x_A\|^2-(1/2)\|\vec x_a\|^2
      \leq \|\vec x_A\|^2-(1/2)\alpha^{6\ell \kappa}.
    \end{multline*}
    On the other hand, since $p$ is a~path of $n$-trees, we get from the
    previous lemma that \( \|\vec x_B\|^2 \geq \|\vec x_A\|^2  -
    \alpha^{(6\ell +2)\kappa} \). If we adjust constant $c_0$ from
    Section~\ref{sec:preproc} so that $1/2>  \alpha^{2\kappa}$, the above
    inequalities give a~contradiction.
  \end{proof}

  The final lemma of this section proves a~slight generalization of the condition $\IPQ{n}$.

  \begin{lemma}\label{lem:consist}
  Let $x$ be a~variable, let $p$ and $q$ be two patterns from $x$ to $x$
  which are paths of $n$-trees, let $1 \leq \ell \leq |D|$, and let
  $A\subseteq D_x^\ell$. Then there exists some $j$ such that $A \subseteq
  A +^\ell (j(p + q) + p)$.
  \end{lemma}

  \begin{proof} For every $A$, define $A_0,A_1,\dots$ in the following way.
  If $i = 2j$ is even then $A_i = A +^\ell (j(p+q))$. Otherwise, if $i =
  2j+1$ is odd then $A_i = A +^\ell (j(p+q) + p)$.

  We claim that for every sufficiently large $u$, we have $A_u \cap D_x^\ell =
  A_{u+1} \cap D_x^\ell$.  From the finiteness of $D$, we get that for
  every sufficiently large $u$ there is $u' > u$ such that $A_u = A_{u'}$.
  It follows that there exists some path of $n$-trees pattern $p'$ starting
  and ending in $x$ such that $A_u = A_{u+1} +^\ell p'$.  To prove the claim
  we will show that $\vec x_{A_u}$ and $\vec x_{A_{u+1}}$ are not
  $\ell$-cut. Then the claim follows as otherwise we would have removed all
  constraints involving $x$ in step 3.

  Consider the path $x_1,\dots,x_k$ in $p'$ which connects the beginning
  and end vertices. Further, let $R_i = R$ if the $i$-th edge of the path
  is labeled by $((x_i,x_{i+1}),R)$, and let $R_i = R^{-1}$ if the $i$-th
  edge is labeled by $((x_{i+1},x_i),R)$.  Now define a~sequence
  $B_1,B_2',B_2,\dots,B_m$ inductively by setting $B_1 = A_{u+1}$,
  $B'_{i+1} = B_i +^\ell (x_i,R_i,x_{i+1})$. Further, if $x_{i+1}$ is not
  a~branching vertex, put $B_{i+1} = B_{i+1}'$.  If $x_{i+1}$ is
  a~branching vertex, then let $\Phi_i$ be the set of all subtrees
  separated by $x_{i+1}$ in $p'$, excluding the two such subtrees
  containing the beginning and the end of $p'$. Then, turn each subtree
  in $\Phi_i$ into a pattern by choosing $x_{i+1}$ as beginning and any
  other leaf as end, and define
  \(
    B_{i+1} = \{ b\in B_{i+1}' : \{b\} +^\ell t \neq \emptyset
      \mbox{ for all } t \in \Phi_i \}.
  \)
  As in Lemma~\ref{lem:14}, we know that the sum of the numbers of leaves
  of the trees from $\Phi_i$ that are also leaves of $p'$ is less than
  $n-1$.  Finally, if $\vec x_{A_u}$ are $\vec x_{A_{u+1}}$ are
  $\ell$-cut then, for some $i$, vectors ${\vec x_i}_{B_i}$ and ${\vec
  x_{i+1}}_{B'_{i+1}}$ are $\ell$-cut, or vectors ${\vec x_i}_{B_i}$ and
  ${\vec x_i}_{B_i'}$ are $\ell$-cut. The former case is impossible since
  $B'_{i+1} = B_i +^\ell (x_i,R_i,x_{i+1})$, and hence if $\vec
  x_{B'_{i+1}}$ and $\vec x_{B_i}$ are $\ell$-cut, then either of the
  constraints $((x_i,x_{i+1}), R_i)$ or $((x_{i+1},x_i),R^{-1})$ would
  have been removed in step 4. We now show that the latter case is
  impossible either. Clearly, in this case $x_i$ is a branching vertex.
  For $t\in \Phi_i$, let $C_t =\{ b \in B_i' : \{b\} +^\ell t = \emptyset
  \}$ and let $j_t$ be the number of leaves of $t$. By
  Lemma~\ref{lem:lost-on-a-tree} we get \( \| {\vec x_i}_{C_t} \|^2 \leq
  (2j_t -1) \alpha^{(6\ell + 4)\kappa} \) for any $t \in \Phi_i$, and
  consequently,
  \[
    \| {\vec x_i}_{B_i'} - {\vec x_i}_{B_i} \|^2 \leq
    \sum_{t\in \Phi_i} \| {\vec x_i}_{C_t} \|^2 \leq 
    \sum_{t\in \Phi_i} (2j_t - 1) \alpha^{(6\ell +4)\kappa} \leq
    (2n - 3) \alpha^{(6\ell +4)\kappa}.
  \]
  Therefore, if ${\vec x_i}_{B_i}$ and ${\vec x_i}_{B_i'}$ were
  $\ell$-cut, then all constraints that include $x_i$ would have been
  removed in step 5.  We conclude that indeed we have $A_u \cap D_x^\ell =
  A_{u+1} \cap D_x^\ell$ for all sufficiently large $u$.

  Now, take $u = 2j+1$ large enough. We have that
  \(
    (A \cup A_{u+1}) +^\ell (j(p+q) + p) = A_u \cup A_{2u+1}.
  \)
  And also $(A_u \cup A_{2u+1}) \cap D_x^\ell = A_{u+1} \cap D_x^\ell
  \subseteq (A \cup A_{u+1})\cap D_x^\ell$, hence by Lemma \ref{lem:maj19}
  we get that
  \( (A \cup A_{u+1}) \cap D_x^\ell = A_{u+1} \cap D_x^\ell \).
  Since $A\subseteq D_x^\ell$ by assumption of the lemma, we have
  \(
    A \subseteq A_{u+1} \cap D_x^\ell \subseteq A_{u} = A +^\ell (j(p+q) + p)
  \).
  \end{proof}

Finally, setting $A=\{a\}$ in Lemma~\ref{lem:consist} gives Proposition~\ref{prop:consistence}.

\section{Full proof of Theorem~\ref{the:main}(2)}\label{sec:thm2}
In this section, we prove Theorem~\ref{the:main}(2). A brief outline of the proof is given in Section~\ref{sec:overview-thm2}. 
Throughout this section, $\calI=(V,\calC)$ is a $(1-\eps)$-satisfiable instance of $\CSP\Gamma$ where $\Gamma$ consists of implicational constraints.

\subsection{SDP Relaxation}\label{sec:thm2:SDP}
We use SDP relaxation (\ref{sdpobj})--(\ref{sdp4}) from Section~\ref{sec:SDP}.
For convenience, we write the SDP objective function as follows. 
\begin{multline}
\sum_{C \in \calC\text{ equals } (x = a) \vee (y = b)}
  w_C \vprod{(\mathbf{v}_0 - \mathbf{x}_{a})}{(\mathbf{v}_0 - \mathbf{y}_{b})} \\
  +\frac{1}{2}\sum_{C \in \calC\text{ equals } x = \pi(y)}
    \,\sum_{a\in D} w_C \| \mathbf{x}_{\pi(a)} - \mathbf{y}_{a}\|^2 \\
  {}+ \sum_{C \in \calC\text{ equals }x \in P} w_C \left(\sum_{a\in D\setminus P} \|\mathbf{x}_{a}\|^2 \right).
\label{SDP} 
\end{multline}
This expression equals (\ref{sdpobj}) because of SDP constraint~(\ref{sdp3}).
\iffalse
Minimize
\begin{multline}
\sum_{C \in \calC\text{ equals } (x = a) \vee (y = b)}
  w_C \vprod{(\mathbf{v}_0 - \mathbf{x}_{a})}{(\mathbf{v}_0 - \mathbf{y}_{b})} \\
  {}+\frac{1}{2}\sum_{C \in \calC\text{ equals } x = \pi(y)}
    \,\sum_{a\in D} w_C \| \mathbf{x}_{\pi(a)} - \mathbf{y}_{a}\|^2 \\
  {}+ \sum_{C \in \calC\text{ equals }x \in P} w_C \left(\sum_{a\in D\setminus P} \|\mathbf{x}_{a}\|^2 \right)
\label{SDP} 
\end{multline}
subject to
\begin{align}
&\vprod{ \mathbf{x}_{a}}{ \mathbf{y}_{b}} \geq 0
   & x,y\in V,\ a,b\in D \label{sdp:pos} \\
&\vprod{ \mathbf{x}_{a}}{ \mathbf{x}_{b}} = 0
   & x\in V,\  a,b\in D,\  a\neq b \\
&\sum_{a\in D} \mathbf{x}_{a} = \mathbf{v}_0
   & x\in V \label{sdp:sum-to-one}\\
&\|\mathbf{x}_{a} - \mathbf{z}_{c}\|^2 \leq \|\mathbf{x}_{a} -
\mathbf{y}_{b}\|^2 + \|\mathbf{y}_{b} - \mathbf{z}_{c}\|^2 \hskip-8em\label{sdp:triangle-ineq}\\
   && x,y,z\in V,\ a,b,c\in D\notag\\
&\|\mathbf{v}_0\|^2 = 1.\label{sdp:unitnorm}
\end{align}

This SDP and the SDP we presented in Section~\ref{sec:SDP} are almost identical. Their objective functions are equal, because of constraints~(\ref{sdp3}) and~(\ref{sdp:sum-to-one}). For convenience, we write the objective function differently in SDP (\ref{SDP})--(\ref{sdp:unitnorm}). 
The only difference between the SDPs is the presence of the ``triangle inequalities'' (\ref{sdp:triangle-ineq}). We introduce them, because
we use the algorithm from~\cite{Charikar06:near} for Unique Games,
which assumes that the SDP has triangle inequalities.
\fi

As discussed before (Lemma \ref{le:prep1}) we can assume
that $\eps\geq 1/m^2$ where $m$ is the number of constraints. We solve SDP
with error $\delta=1/m^2$ obtaining a solution, denoted $\mathsf{SDP}$, with objective value $O(\epsilon)$.
Note that every feasible SDP solution satisfies the following conditions.
\begin{align}
  &\| \mathbf{x}_{a} \|^2 = \vproddot{ \mathbf{x}_{a}}{\bigl(\mathbf{v}_0 - \sum_{b\neq a} \mathbf{x}_{b}\bigr)} = \vproddot{ \mathbf{x}_{a}}{\mathbf{v}_0}-\sum_{b\neq a}\vproddot{\mathbf{x}_a}{\mathbf{x}_b}
    = \vprod{ \mathbf{x}_{a}}{ \mathbf{v}_0 },\label{eq:length}\\[1mm]
  &\vprod{ \mathbf{x}_{a} }{ \mathbf{y}_{b}}
    = \vproddot{ \mathbf{x}_{a}}{(\mathbf{v}_0 - \sum_{b'\neq b}\mathbf{y}_{b'})} \label{eq:2SATrequirement} = \|\mathbf{x}_{a}\|^2 - \sum_{b'\neq b}
      \vprod{ \mathbf{x}_{a}}{ \mathbf{y}_{b'}} 
    \leq
      \|\mathbf{x}_{a}\|^2, \\[1mm]
&  \|\mathbf{x}_{a}\|^2 - \|\mathbf{y}_{b}\|^2
    = \|\mathbf{x}_{a} - \mathbf{y}_{b}\|^2 +
      2(\vprod{\mathbf{x}_{a}}{\mathbf{y}_{b}} - \|\mathbf{y}_{b}\|^2)
       \leq \|\mathbf{x}_{a} - \mathbf{y}_{b}\|^2, \label{eq:triangle}\\[1mm]
&  \vprod{(\mathbf{v}_0 - \mathbf{x}_{a})}{( \mathbf{v}_0 - \mathbf{y}_{b})}
    = \vprod{ \sum_{a'\neq a} \mathbf{x}_{a'}}
      {\sum_{b'\neq b} \mathbf{y}_{b'}} \geq 0.\label{eq:positivity}
\end{align}
\subsection{Variable Partitioning Step}\label{sec:2pre}
In this section, we describe the first step of our algorithm. In this step, we assign values to some variables, partition all variables into three groups
${\cal V}_0$, ${\cal V}_1$ and ${\cal V}_2$,
and then split the instance into two sub-instances ${\cal I}_1$ and ${\cal I}_2$.

\paragraph{Vertex Partitioning Procedure.}

\noindent Choose a number $r \in (0, 1/6)$ uniformly at random. Do the following for every variable $x$.
\begin{enumerate}
\item Let $D_x = \{a: 1/2 -r < \vprod{ \mathbf{x}_{a}}{ \mathbf{v}_0}\}$.
\item Depending on the size of $D_x$ do the following:
\begin{enumerate}
\item If $|D_x| = 1$, add $x$ to ${\cal V}_0$ and  assign $x = a$, where $a$ is the single element of $D_x$.
\item If $|D_x| > 1$, add $x$ to ${\cal V}_1$ and restrict $x$ to $D_x$ (see below for details).
\item If $D_x = \varnothing$, add $x$ to ${\cal V}_2$.
\end{enumerate}
\end{enumerate}

Note that each variable in ${\cal V}_0$ is assigned a value; each variable $x$ in ${\cal V}_1$ is restricted to a set $D_x$;
each variable in ${\cal V}_2$ is not restricted.

\begin{lemma}
(i) If $\vprod{ \mathbf{x}_{a}}{ \mathbf{v}_0} > \frac{1}{2} + r$ then $x\in {\cal V}_0$.
(ii) For every $x\in {\cal V}_1$, $|D_x| = 2$.
\end{lemma}
\begin{proof}
(i) Note that for every $b\neq a$, we have $\vprod{ \mathbf{x}_{a}}{ \mathbf{v}_0}  + \vprod{ \mathbf{x}_{b}}{ \mathbf{v}_0} \leq 1$ and,
therefore, $\vprod{\mathbf{x}_{b}}{\mathbf{v}_0} < 1/2 -r$. Hence, $b\notin D_x$. We conclude that $D_x = \{a\}$ and $x\in {\cal V}_0$.

\smallskip
\noindent(ii) Now consider $x \in {\cal V}_1$. We have,
$$
  |D_x| <  3(1/2 -r) |D_x| = 
    3\sum_{a\in D_x} (1/2 -r)
    \leq 3\sum_{a\in D_x} \vprod{ \mathbf{x}_{a}}{ \mathbf{v}_0}  \leq 3.
$$
Therefore,  $|D_x| \leq 2$. Since $x \in {\cal V}_1$, $|D_x| > 1$. Hence $|D_x| = 2$.
\end{proof}

We say that an assignment is admissible if it assigns a value in $D_x$ to every $x \in {\cal V}_1$ and it is consistent with the partial assignment to variables in ${\cal V}_0$.
From now on we restrict our attention only to admissible assignments. We remove those constraints
that are satisfied by every admissible assignment (our algorithm will satisfy all of them). Specifically,
we remove the following constraints:
\begin{enumerate}
\item UG constraints $x=\pi(y)$ with $x, y \in {\cal V}_0$ that are satisfied by the partial assignment;
\item disjunction constraints $(x = a) \vee (y = b)$ such that either $x \in {\cal V}_0$ and $x$ is assigned value $a$, or $y \in {\cal V}_0$ and
$y$ is assigned value $b$;
\item unary constraints $x \in P$ such that either $x \in {\cal V}_0$ and the value assigned to $x$ is in $P$, or $x \in {\cal V}_1$ and $D_x \subseteq P$.
\end{enumerate}
We denote the set of satisfied constraints by ${\cal C}_s$.
Let ${\cal C}'={\cal C}\setminus {\cal C}_s$ be the set of remaining constraints. We now define a set of \textit{violated} constraints --- those constraints that we conservatively assume will not be
satisfied by our algorithm (even though some of them might be satisfied by the algorithm).
We say that a constraint $C\in {\cal C}'$ is violated if at least one of the following conditions holds:
\begin{enumerate}
\item $C$ is a unary constraint on a variable $x \in {\cal V}_0 \cup {\cal V}_1$.
\item $C$ is a disjunction constraint $(x = a) \vee (y = b)$ and either $x \notin {\cal V}_1$, or $y \notin {\cal V}_1$ (or both).
\item $C$ is a disjunction constraint $(x = a) \vee (y = b)$, and $x, y \in {\cal V}_1$, and either $a\notin D_x$, or $b\notin D_y$ (or both).
\item $C$ is a UG constraint $x = \pi(y)$, and at least one of the variables $x$, $y$ is in ${\cal V}_0$.
\item $C$ is a UG constraint $x = \pi(y)$, and one of the variables $x$, $y$ is in ${\cal V}_1$ and the other is in ${\cal V}_2$.
\item $C$ is a UG constraint $x = \pi(y)$, $x, y\in {\cal V}_1$ but $D_x \neq \pi(D_y)$.
\end{enumerate}
We denote the set of violated constraints by ${\cal C}_v$ and let ${\cal C}'' = {\cal C}'\setminus {\cal C}_v$.
\begin{lemma}\label{lem:preproc-violated}
$\Exp[w(\calC_v)] = O(\eps)$.
\end{lemma}
\begin{proof}
We analyze separately constraints of each type in $\calC_v$.

\subsubsection*{Unary constraints}
A unary constraint $x \in P$ in $\cal C$ is violated if and only if $x \in {\cal V}_0 \cup {\cal V}_1$ and $D_x \not\subseteq P$
(if $D_x \subseteq P$ then $C\in {\cal C}_s$ and thus $C$ is not violated).
Thus the SDP contribution of each violated constraint $C$ of the form $x \in P$ is at least
$$
  w_C\sum_{a\in D\setminus P} \|\mathbf{x}_{a}\|^2
  \geq w_C\sum_{a\in D_x\setminus P}\|\mathbf{x}_{a}\|^2 
  = w_C\sum_{a\in D_x\setminus P} \vproddot{ \mathbf{x}_{a}}{ \mathbf{v}_0}
  \geq w_C\Bigl(\frac{1}{2} - r\Bigr)\geq \frac{w_C}{3}.
$$
The last two inequalities hold because the set $D_x\setminus P$ is nonempty;
$\mathbf{x}_{a} \mathbf{v}_0 \geq 1/2-r$ for all $a\in D_x$ by the construction; and $r\leq 1/6$.
Therefore, the expected total weight of violated unary constraints is at most $3\,\mathsf{SDP}= O(\eps)$.

\subsubsection*{Disjunction constraints}
Consider a disjunction constraint $(x = a) \vee (y = b)$.  Denote it by $C$.
Assume without loss of generality that $\vprod{ \mathbf{x}_{a}}{ \mathbf{v}_0 } \geq  \vprod{ \mathbf{y}_{b}}{ \mathbf{v}_0}$.
Consider several cases. If $\vprod{ \mathbf{x}_{a}}{ \mathbf{v}_0 } > 1/2 +r$ then $x \in {\cal V}_0$ and $x$ is assigned value $a$. Thus, $C$ is satisfied.
If $\vprod{ \mathbf{x}_{a}}{ \mathbf{v}_0 } \leq 1/2 +r$ and $\vprod{ \mathbf{y}_{b}}{ \mathbf{v}_0} > 1/2 -r$ then we
also have $\vprod{ \mathbf{x}_{a}}{ \mathbf{v}_0} > 1/2 - r$ and hence $x, y\in {\cal V}_0 \cup {\cal V}_1$ and $a\in D_x$, $b\in D_y$.
Thus, $C$ is not violated (if at least one of the variables $x$ and $y$ is in ${\cal V}_0$, then $C\in {\cal C}_s$; otherwise, $C \in {\cal C}'$). Therefore, $C$ is violated only if
$$\vprod{ \mathbf{x}_{a}}{ \mathbf{v}_0 } \leq 1/2 +r \text{ and }  \vprod{ \mathbf{y}_{b}}{ \mathbf{v}_0} \leq 1/2 -r ,$$
or equivalently,
\begin{equation}
   \vprod{ \mathbf{x}_{a}}{ \mathbf{v}_0 } - 1/2 \leq r \leq 1/2 - \vprod{ \mathbf{y}_{b}}{ \mathbf{v}_0}.\label{eq:bad-event}
\end{equation}
Since we choose $r$ uniformly at random in $(0,1/6)$, the probability density of the random variable $r$ is 6 on $(0,1/6)$. Thus the probability of event (\ref{eq:bad-event}) is at most
\begin{multline*}
  6 \max\Bigl(\bigl(( 1/2 - \vprod{ \mathbf{y}_{b}}{ \mathbf{v}_0}\bigr)-
    \bigl(\vprod{ \mathbf{x}_{a}}{ \mathbf{v}_0 } - 1/2)\bigr),0\Bigr) \\
  = 6\max\Bigl( \vprod{(\mathbf{v}_0 - \mathbf{x}_{a})}{(\mathbf{v}_0 - \mathbf{y}_{b})}
    -\vprod{\mathbf{x}_{a}}{\mathbf{y}_{b}},0\Bigr) \\
  {}\stackrel{\text{by (\ref{sdp1}) and (\ref{eq:positivity})}}{\leq}
    6\vprod{(\mathbf{v}_0 - \mathbf{x}_{a})}{(\mathbf{v}_0 - \mathbf{y}_{b})}.
\end{multline*}
The expected weight of violated constraints is at most,
$$ \sum_{\substack{C\in\calC  \text{ equals }\\(x = a) \vee (y = b)}}
 6w_C\vprod{(\mathbf{v}_0 - \mathbf{x}_{a})}{(\mathbf{v}_0 - \mathbf{y}_{b})}
\leq 6\,\mathsf{SDP} = O(\eps). $$

\subsubsection*{UG constraints} Consider a UG constraint $x = \pi(y)$. Assume that it is violated.
Then $D_x \neq \pi(D_y)$ (note that if $x$ and $y$ do not lie in the same set ${\cal V}_t$ then $|D_x| \neq |D_y|$ and
necessarily $D_x \neq \pi(D_y)$).
Thus, at least one of the sets $\pi(D_y) \setminus D_x$ or $D_x \setminus \pi(D_y)$ is not empty.
If  $\pi(D_y) \setminus D_x\neq \varnothing$, there exists $c \in \pi(D_y) \setminus D_x$. We have,
\begin{align*}
\Prob{c \in \pi(D_y) \setminus D_x} 
&\leq \Prob{\|\mathbf{y}_{\pi^{-1}(c)}\|^2 > 1/2 - r \text{ and } \|\mathbf{x}_{c}\|^2 \leq 1/2 - r} \\
&= \Prob{1/2 - \|\mathbf{y}_{\pi^{-1}(c)}\|^2 <  r \leq 1/2 -\|\mathbf{x}_{c}\|^2}
\\
&\leq 6\max(\|\mathbf{y}_{\pi^{-1}(c)}\|^2 - \|\mathbf{x}_{c}\|^2, 0) 
\stackrel{\text{by (\ref{eq:triangle})}}{\leq} 6\|\mathbf{y}_{\pi^{-1}(c)} - \mathbf{x}_{c}\|^2.
\end{align*}
By the union bound, the probability that there is $c \in \pi(D_y) \setminus D_x$ is at most
\[
  6\sum_{c\in D}\|\mathbf{y}_{\pi^{-1}(c)} - \mathbf{x}_{c}\|^2=6\sum_{b\in D}\|\mathbf{y}_{b} - \mathbf{x}_{\pi(b)}\|^2
.\]
Similarly, the probability that there is $b \in D_x \setminus \pi(D_y)$ is at most $6\sum_{b\in D}\|\mathbf{y}_{b} - \mathbf{x}_{\pi(b)}\|^2$. 
Therefore, the probability that the constraint $x=\pi(y)$ is violated is 
upper bounded by $12\sum_{b\in D}\|\mathbf{y}_{b} - \mathbf{x}_{\pi(b)}\|^2$. Consequently, the total expected weight of all violated UG constraints is at most
\begin{multline*}\sum_{C \in \calC\text{ equals } x = \pi(y)} w_C \left(12 
    \,\sum_{b\in D} \| \mathbf{x}_{\pi(b)} - \mathbf{y}_{b}\|^2\right) \\
    =
24 \times \left(\frac12 \sum_{C \in \calC\text{ equals } x = \pi(y)} w_C  
    \,\sum_{b\in D} \| \mathbf{x}_{\pi(b)} - \mathbf{y}_{b}\|^2\right)
    \leq 24\,\mathsf{SDP} = O(\eps),
\end{multline*}    
here we bound the value of the SDP by the second term of the objective 
function~(\ref{SDP}).
\end{proof}

\noindent We restrict our attention to the set ${\cal C}''$. There are four types of constraints in ${\cal C}''$.
\begin{enumerate}
  \item disjunction constraints $(x = a) \vee (y = b)$ with $x, y \in {\cal V}_1$ and $a\in D_x$, $b\in D_y$;
  \item UG constraints $x = \pi (y)$ with $x, y \in {\cal V}_1$ and $D_x = \pi(D_y)$;
  \item UG constraints $x = \pi (y)$ with $x, y \in {\cal V}_2$;
  \item unary constraints $x \in P$ with $x\in{\cal V}_2$.
\end{enumerate}
Denote the set of type 1 and 2 constraints by ${\cal C}_1$, and type 3 and 4 constraints by ${\cal C}_2$.
Let ${\cal I}_1$ be the sub-instance of ${\cal I}$ on variables ${\cal V}_1$ with constraints ${\cal C}_1$ in which every variable
$x$ is restricted to $D_x$, and ${\cal I}_2$ be the sub-instance of ${\cal I}$ on variables ${\cal V}_2$ with constraints ${\cal C}_2$.

In Sections~\ref{solve-i1} and~\ref{solve-i2}, we  show how to solve ${\cal I}_1$ and ${\cal I}_2$, respectively. The total weight of constraints violated by our solution
for ${\cal I}_1$  will be at most $O(\sqrt{\eps})$;
 The total weight of constraints violated by our solution for ${\cal I}_2$  will be at most
$O(\sqrt{\eps\log{|D|} })$.
Thus the combined solution will satisfy a subset of the constraints of weight at least $1 - O(\sqrt{\eps \log{|D|}})$.

\subsection{Solving Instance \texorpdfstring{${\cal I}_1$}{I1}}\label{solve-i1}

In this section, we present an algorithm that solves instance ${\cal I}_1$. The algorithm assigns
values to variables in ${\cal V}_1$ so that the total weight of violated constraints is at most $O(\sqrt{\eps})$.

\begin{lemma}\label{lem:consistentUG}
There is a randomized algorithm that, given instance ${\cal I}_1$ and the SDP solution $\{\bx_a\}$ for ${\cal I}$, finds a set of UG constraints ${\cal C}_{\text{bad}}\subseteq {\cal C}_1$  and values $\alpha_x,\beta_x \in D_x$ for every $x\in {\cal V}_1$
such that the following conditions hold.
\begin{itemize}
\item $D_x = \{\alpha_x, \beta_x\}$.
\item for each UG constraint $x=\pi(y)$ in ${\cal C}_1\setminus {\cal C}_{\text{bad}}$, we have $\alpha_x =\pi(\alpha_y)$ and $\beta_x =\pi(\beta_y)$.
\item 
The expected weight of ${\cal C}_{\text{bad}}$ is $O(\sqrt{\eps}).$
\end{itemize}
\end{lemma}
\begin{proof}
We use the algorithm of Goemans and Williamson for Min Uncut~\cite{Goemans95:improved} to find values~$\alpha_x$, $\beta_x$.
Recall that in the Min Uncut problem (also known as Min 2CNF$\equiv$ deletion)  we are given a set of Boolean variables and a set of constraints of the form
$(x = a) \leftrightarrow (y = b)$. Our goal is to find an assignment that minimizes the weight of unsatisfied
constraints.

Consider the set of UG constraints in ${\cal C}_1$. Since $|D_x| = 2$ for every variable $x \in {\cal V}_1$,
each constraint $x = \pi(y)$ is equivalent to the Min Uncut constraint $(x =\pi(a)) \leftrightarrow (y =a)$
where $a$ is an element of $D_y$ (it does not matter which of the two elements of $D_y$ we choose). We define
an SDP solution for the Goemans---Williamson relaxation of Min Uncut as follows.
Consider $x \in {\cal V}_1$. Denote the elements of $D_x$ by $a$ and $b$ (in any order). Let
$$\mathbf{x}^*_{a} = \frac{\mathbf{x}_{a}-\mathbf{x}_{b}}{\|\mathbf{x}_{a}-\mathbf{x}_{b}\|} \quad\text{and}\quad \mathbf{x}^*_{b} = - \mathbf{x}^*_{a} = \frac{\mathbf{x}_{b}-\mathbf{x}_{a}}{\|\mathbf{x}_{a}-\mathbf{x}_{b}\|}.$$
Note that the vectors $\mathbf{x}_a$ and $\mathbf{x}_b$ are nonzero orthogonal vectors, and, thus,
$\|\mathbf{x}_a - \mathbf{x}_b\|$ is nonzero. The vectors $\mathbf{x}^*_{a}$ and $\mathbf{x}^*_{b}$ are unit vectors. Now we apply the random hyperplane rounding scheme of
Goemans and Williamson:
We choose a random hyperplane and let $H$ be one of the half-spaces the hyperplane divides the space into. Note that
for every $x$ exactly one of the two antipodal vectors in $\{\mathbf{x}^*_{a}:a\in D_x\}$ lies in $H$ (almost surely). Define $\alpha_x$ and $\beta_x$
so that $\mathbf{x}^*_{\alpha_x}\in H$ and $\mathbf{x}^*_{\beta_x}\notin H$. Let ${\cal C}_{\text{bad}}$ be the set of UG constraints such that
$\alpha_x \neq \pi(\alpha_y)$, or equivalently $\mathbf{x}_{\pi(\alpha_y)}^* \notin H$.

Values $\alpha_x$ and $\beta_x$ satisfy the first condition. If a UG constraint
$x=\pi(y)$ is in ${\cal C}_1\setminus {\cal C}_{\text{bad}}$, then $\alpha_x =
\pi(\alpha_y)$; also since $D_x = \pi(D_y)$, $\beta_x = \pi(\beta_y)$. So the
second condition holds.  Finally, we verify the last condition. 
Consider a
constraint  $x = \pi(y)$. Let $\mathbf{A} = \mathbf{x}_{\pi(\alpha_y)} -
\mathbf{x}_{\pi(\beta_y)}$ and $\mathbf{B} = \mathbf{y}_{\alpha_y} -
\mathbf{y}_{\beta_y}$. Since $x\in {\cal V}_1$, we have
$\|\mathbf{x}_{\pi(\alpha_y)}\|^2 > 1/2 -r > 1/3$ and
$\strut\|\mathbf{x}_{\pi(\beta_y)}\|^2 > 1/3$.  Hence $\|\mathbf{A}\|^2 =
\|\mathbf{x}_{\pi(\alpha_y)}\|^2 + \|\mathbf{x}_{\pi(\beta_y)}\|^2  > 2/3$.
Similarly, $\|\mathbf{B}\|^2> 2/3$. Assume first that $\|\mathbf{A}\| \geq \|\mathbf{B}|$.
Then,
\begin{multline*}
\|\mathbf{x}^*_{\pi(\alpha_y)} - \mathbf{y}^*_{\alpha_y}\|^2 =
\left\|\frac{\mathbf{A}}{\|\mathbf{A}\|} -
\frac{\mathbf{B}}{\|\mathbf{B}\|}\right\|^2  \\
= 2 - \frac{2\vprod{\mathbf{A}}{\mathbf{B}}}{\|\mathbf{A}\| \|\mathbf{B}\|}
=\frac{2}{\|\mathbf{B}\|^2}
\times \left(
 \|\mathbf{B}\|^2 - 
  \frac{\|\mathbf{B}\|}{\|\mathbf{A}\|}
  \, \vprod{\mathbf{A}}{\mathbf{B}}\right).
\end{multline*}
We have $2\Bigl(\|\mathbf{B}\|^2 - \frac{\|\mathbf{B}\|}{\|\mathbf{A}\|}\,\vprod{\mathbf{A}}{\mathbf{B}} \Bigr) \leq \|\mathbf{A} -\mathbf{B}\|^2$, since
$$
\|\mathbf{A} -\mathbf{B}\|^2 - 2\Bigl(\|\mathbf{B}\|^2 - \frac{\|\mathbf{B}\|}{\|\mathbf{A}\|}\,\vprod{\mathbf{A}}{\mathbf{B}} \Bigr)
=
\Bigl(\|\mathbf{A}\|-\|\mathbf{B}\|\Bigr)\Bigl(\|\mathbf{A}\|+\|\mathbf{B}\| - \frac{2 \vprod{\mathbf{A}}{\mathbf{B}}}{\|\mathbf{A}\|}\Bigr) \geq 0,
$$
because $\|\mathbf{A}\| \geq \vprod{\mathbf{A}}{\mathbf{B}}/\|\mathbf{A}\|$ and $\|\mathbf{B}\| \geq \vprod{\mathbf{A}}{\mathbf{B}}/\|\mathbf{A}\|$.
We conclude that
\begin{multline*}
\|\mathbf{x}^*_{\pi(\alpha_y)} - \mathbf{y}^*_{\alpha_y}\|^2 
  \leq \frac{\|\mathbf{A}-\mathbf{B}\|^2}{\|\mathbf{B}\|^2} \leq \frac{3}{2}\|\mathbf{A}-\mathbf{B}\|^2 \\
= \frac{3}{2}\, \|(\mathbf{x}_{\pi(\alpha_y)} -\mathbf{y}_{\alpha_y}) -
  (\mathbf{x}_{\pi(\beta_y)} - \mathbf{y}_{\beta_y}) \|^2 \\
\leq 3\, \|\mathbf{x}_{\pi(\alpha_y)} -\mathbf{y}_{\alpha_y}\|^2 + 3\,
  \|\mathbf{x}_{\pi(\beta_y)} - \mathbf{y}_{\beta_y}\|^2.
\end{multline*}
If $\|\mathbf{A}\|\leq \|\mathbf{B}\|$, we get the same bound on 
$\|\mathbf{x}^*_{\pi(\alpha_y)} - \mathbf{y}^*_{\alpha_y}\|^2$
by swapping $\mathbf{A}$ and $\mathbf{B}$ in the formulas above.
Therefore,
\[
  \sum_{\substack{C\in{\cal C}_1 \\ \text{ is of the form } \\ x = \pi(y)}}
  w_C\|\mathbf{x}^*_{\pi(\alpha_y)} - \mathbf{y}^*_{\alpha_y}\|^2\leq 3\,
  \mathsf{SDP} = O(\eps).
\]
The analysis by Goemans and Williamson shows that the expected total weight of the constraints of the form  $x = \pi(y)$ such that
$$\mathbf{x}^*_{\pi(\alpha_y)}\notin H \text{ and } \mathbf{y}^*_{\alpha_y}\in H$$
is at most $O(\sqrt{\eps})$, see Section 3 in~\cite{Goemans95:improved} for the original analysis or Section 2 in survey~\cite{MM17:cspsurvey} for presentation more closely aligned with our notation.
Therefore, the expected total weight of ${\cal C}_{\text{bad}}$ is  $O(\sqrt{\eps})$.
\end{proof}

We remove all constraints ${\cal C}_{\text{bad}}$ from ${\cal I}_1$ and obtain an instance ${\cal I}_1'$ (with the domain for each variable $x$ now restricted to $D_x$).
We  construct an SDP solution $\{\mathbf{\tilde x}_{a}\}$ for ${\cal I}_1'$. We let
$$\mathbf{\tilde x}_{\alpha_x} = \mathbf{x}_{\alpha_x} \quad\text{and}\quad \mathbf{\tilde x}_{\beta_x} = \mathbf{v}_0 - \mathbf{x}_{\alpha_x}.$$
We define $S_{x\alpha_x} = \{\alpha_x\}$ and $S_{x\beta_x} = D\setminus S_{x\alpha_x}$. Since $\mathbf{\mathbf{\tilde x}}_{\beta_x} = \mathbf{v}_0 - \mathbf{x}_{\alpha_x} = 
\mathbf{x}_{S_{x\beta_x}}$,
we have,
\begin{equation}
\mathbf{\tilde x}_{a} = \mathbf{x}_{S_{xa}} 
\quad \text{for every } a\in D_x.\label{eq:redef-x-a}
\end{equation}
 Note that $a\in S_{xa}$ for every $a\in D_x$.

\begin{lemma}\label{lem:new-SDP-feasible}
The solution $\{\mathbf{\tilde x}_{a}\}$ is a feasible solution for SDP relaxation (\ref{sdpobj})--(\ref{sdp4}) for ${\cal I}'_1$. Its cost is $O(\eps)$.
\end{lemma}
\begin{proof}
We verify that the SDP solution is feasible.
First, we have $\sum_{a\in D_x} \mathbf{\tilde x}_{a} = \mathbf{v}_0$ and
$$\vprod{ \mathbf{\tilde x}_{\alpha_x}}{ \mathbf{\tilde x}_{\beta_x}} =
 \vproddot{ \mathbf{x}_{\alpha_x}}{ (\mathbf{v}_0 - \mathbf{x}_{\alpha_x})} = \vprod{ \mathbf{x}_{\alpha_x}}{ \mathbf{v}_0 } - \|\mathbf{x}_{\alpha_x}\|^2 = 0.$$
Then for $a\in D_x$ and $b\in D_y$, we have
$\vprod{ \mathbf{\tilde x}_{a}}{ \mathbf{\tilde y}_{b}} = \sum_{a'\in S_{xa}, b'\in S_{yb}} \vprod{ \mathbf{x}_{a'}}{ \mathbf{y}_{b'}} \geq 0$.
We now show that the SDP cost is $O(\eps)$.

First, we consider disjunction constraints.
We prove that the contribution of each constraint $(x = a) \vee (y =b)$ to the SDP for ${\cal I}_1'$ is at most its contribution to the SDP for $\cal I$.
That is,
\begin{equation}
\vprod{(\mathbf{v}_0 -  \mathbf{\tilde x}_{a})}{(\mathbf{v}_0 - \mathbf{\tilde y}_{b})}
\leq \vprod{ (\mathbf{v}_0 - \mathbf{x}_{a})}{(\mathbf{v}_0 - \mathbf{y}_{b})}. \label{ineq:SDP}
\end{equation}

Observe that
$(\mathbf{v}_0 -  \mathbf{\tilde x}_{a}) =
\mathbf{x}_{D\setminus S_{xa}}$, 
$(\mathbf{v}_0 -  \mathbf{\tilde y}_{b}) =
\mathbf{y}_{D\setminus S_{yb}}$, 
$(\mathbf{v}_0 -  \mathbf{x}_{a}) =
\mathbf{x}_{D\setminus \{a\}}$, and 
$(\mathbf{v}_0 -  \mathbf{y}_{b}) =
\mathbf{y}_{D\setminus \{b\}}$. 
Then, $D\setminus S_{xa}\subseteq D\setminus\{a\}$ and
$D\setminus S_{yb}\subseteq D\setminus\{b\}$. Therefore,
by (\ref{sdp1}),
\begin{multline*}
\vprod{(\mathbf{v}_0 -  \mathbf{\tilde x}_{a})}{(\mathbf{v}_0 - \mathbf{\tilde y}_{b})} 
= 
\sum_{(a',b')\in (D\setminus S_{xa})\times (D\setminus S_{yb})} 
\vprod{\mathbf{x}_{a'}}{\mathbf{y}_{b'}}
\leq \\ \leq 
\sum_{(a',b')\in (D\setminus \{a\})\times (D\setminus 
\{b\})}\vprod{\mathbf{x}_{a'}}{\mathbf{y}_{b'}}
=
\vprod{ (\mathbf{v}_0 - \mathbf{x}_{a})}{(\mathbf{v}_0 - \mathbf{y}_{b})}. 
\end{multline*}

\iffalse
Observe that
\begin{multline*}
\vprod{(\mathbf{v}_0 - \mathbf{x}_{a})}{(\mathbf{v}_0 - \mathbf{y}_{b})}
- \vprod{(\mathbf{v}_0 -  \mathbf{\tilde x}_{a})}{(\mathbf{v}_0 - \mathbf{\tilde y}_{b})}
= \\
\vprod{(\mathbf{v}_0 - \mathbf{\tilde x}_{a})}{(\mathbf{\tilde y}_{b} - \mathbf{y}_{b})}
+
\vprod{{(\mathbf{\tilde x}_{a}- \mathbf{x}_{a})}(\mathbf{v}_0 - \mathbf{\tilde
y}_{b})} \\
{}+
\vprod{(\mathbf{\tilde x}_a - \mathbf{x}_{a})}{(\mathbf{\tilde y}_b - \mathbf{y}_{b})}
.
\end{multline*}
We prove that all terms on the right hand side are nonnegative, and thus inequality (\ref{ineq:SDP}) holds.
Using the identities (\ref{eq:redef-x-a}) and $\sum_{a'\in D} \mathbf{x}_{a'} = \mathbf{v}_{0}$ as well as
the inequality
$\mathbf{x}_{a'}\mathbf{y}_{b'}\geq 0$ (for all $a',b'\in D$), we get
$$
\vprod{(\mathbf{v}_0 - \mathbf{\tilde x}_{a})}{(\mathbf{\tilde y}_{b} - \mathbf{y}_{b})} =
\sum_{\substack{a'\in D\setminus S_{xa}\\b'\in S_{yb}\setminus \{b\}}}  \mathbf{x}_{a'}\mathbf{y}_{b'} \geq 0.
$$
Similarly, $\vprod{{(\mathbf{\tilde x}_{a}- \mathbf{x}_{a})}(\mathbf{v}_0 - \mathbf{\tilde y}_{b})}\geq 0$, and
$$
\vprod{(\mathbf{\tilde x}_a - \mathbf{x}_{a})}{(\mathbf{\tilde y}_b - \mathbf{y}_{b})}=
\sum_{\substack{a'\in S_{xa}\setminus \{a\}\\b'\in S_{yb}\setminus \{b\}}}  \mathbf{x}_{a'}\mathbf{y}_{b'} \geq 0.
$$
\fi
Now we consider UG constraints. The contribution of a UG constraint $x = \pi(y)$ in ${\cal C}_1 \setminus {\cal C}_{\text{bad}}$ to
the SDP for ${\cal I}_1'$ equals the weight of the constraint times the following expression.
\begin{multline*}
\|\mathbf{\tilde x}_{\pi(\alpha_y)} - \mathbf{\tilde y}_{\alpha_y}\|^2 +
\|\mathbf{\tilde x}_{\pi(\beta_y)} - \mathbf{\tilde y}_{\beta_y}\|^2 = 
\|\mathbf{\tilde x}_{\alpha_x} - \mathbf{\tilde y}_{\alpha_y}\|^2 +
\|\mathbf{\tilde x}_{\beta_x} - \mathbf{\tilde y}_{\beta_y}\|^2 = \\
\|\mathbf{x}_{\alpha_x} - \mathbf{y}_{\alpha_y}\|^2 +
\|(\mathbf{v}_0 - \mathbf{x}_{\alpha_x}) -
(\mathbf{v}_0-\mathbf{y}_{\alpha_y})\|^2 = \\
2 \|\mathbf{x}_{\alpha_x} - \mathbf{y}_{\alpha_y}\|^2 = 2 \|\mathbf{x}_{\pi(\alpha_y)} - \mathbf{y}_{\alpha_y}\|^2.
\end{multline*}
Thus, by the choice of $\alpha_x$ and $\alpha_y$ (Lemma \ref{lem:consistentUG}) the contribution is at most twice the contribution of the constraint to the SDP for ${\cal I}$.
We conclude that the SDP contribution of all the constraints in ${\cal C}_1 \setminus {\cal C}_{\text{bad}}$ is at most $2\,\mathsf{SDP} = O(\eps)$.
\end{proof}

Finally, we note that ${\cal I}_1'$ is a Boolean 2-CSP instance. We round solution $\{\mathbf{\tilde x}_{a}\}$ using the rounding procedure by Charikar et al. for Boolean 2-CSP~\cite{Charikar09:near} (when $|D| = 2$, the SDP relaxation used in~\cite{Charikar09:near} is equivalent to SDP (\ref{sdpobj})--(\ref{sdp4})).
We get an assignment of variables in ${\cal V}_1$. The weight of constraints in  ${\cal C}_{1}\setminus {\cal C}_{\text{bad}}$ violated by this assignment is  at most $O(\sqrt{\eps})$.
Since $w({\cal C}_{\text{bad}}) = O(\sqrt{\eps})$, the weight of constraints in ${\cal C}_{1}$ violated by the assignment is at most $O(\sqrt{\eps})$.

\subsection{Solving Instance \texorpdfstring{${\cal I}_2$}{I2}}\label{solve-i2}

Instance ${\cal I}_2$ is a unique games instance with additional unary constraints.
We restrict the SDP solution for $\cal I$ to variables $x \in {\cal V}_2$ and get a solution for the
unique game instance ${\cal I}_2$.  Note that since we do not restrict the domain
of variables $x\in \VV{2}$ to $D_x$, the SDP solution we obtain is feasible.
The SDP cost of this solution is at most $\mathsf{SDP}$.
We round this SDP solution using a variant of the algorithm by Charikar et al.~\cite{Charikar06:near} that is 
presented in Section 3 of the survey~\cite{MM17:cspsurvey};
this variant of the algorithm does not need
$\ell_2^2$-triangle-inequality SDP constraints.
Given a $(1-\eps)$-satisfiable instance of Unique Games, the algorithm finds a solution with the weight of violated constraints at
most $O(\sqrt{\eps \log{|D|}})$.
We remark that paper~\cite{Charikar06:near} considers only unique game instances. However, in~\cite{Charikar06:near},
we can restrict the domain of any variable $x$ to a set $S_x$ by setting $\mathbf{x}_a = 0$ for
$a\in D\setminus S_x$. Hence, we can model unary constraints as follows. For every unary constraint
$x\in P$, we introduce a dummy variable $z_{x,P}$ and restrict its domain to the set $P$. Then we replace
each constraint $x\in P$ with the equivalent constraint $x = z_{x,P}$.
The weight of the constraints violated by the obtained solution is at most  $O(\sqrt{\eps \log{|D|}})$.

Finally, we combine results  proved in Sections~\ref{sec:2pre}, \ref{solve-i1}, and~\ref{solve-i1} and obtain Theorem~\ref{the:main}(2).

\section{Conclusion}
We have proved that every CSP with an NU polymorphism admits a robust algorithm
with polynomial loss. Thus a small gap remains in our understanding of such
algorithms -- between the sufficient condition of having an NU polymorphism and
a necessary condition $\mathrm{SD}(\vee)$. We remark that closing this gap is
likely to require a structural result, similar to our
Theorem~\ref{thm:nicelevels}, which would resolve the conjecture of Larose and
Tesson~\cite{Larose09:universal} and characterise CSPs solvable by linear
propagation. Such a result would immediately imply a characterisation of CSPs
in the complexity class NL~\cite{Dalmau05:linear,Larose09:universal} (and hence
also L~\cite{Kazda18:nperm}), modulo complexity-theoretic assumptions.

  \bibliographystyle{alpha}

\end{document}